\documentclass[manuscript,screen,nonacm]{acmart}

\AtBeginDocument{%
  }

\renewcommand\footnotetextcopyrightpermission[1]{} % no first-page copyright footnote
\usepackage{algorithm}
\usepackage{algpseudocode}
\usepackage{booktabs}
\usepackage{multirow}
\usepackage{graphicx}
\usepackage{microtype}
\usepackage{enumitem}
\usepackage{amsmath,amssymb,amsthm}
\usepackage{xcolor}
\usepackage{url}
\usepackage{xurl}
\usepackage{array}

\algrenewcommand{\algorithmiccomment}[1]{\hfill\textcolor{gray!60!black}{\(\triangleright\)\,\textit{#1}}}

\newtheorem{theorem}{Theorem}

\newtheorem{corollary}[theorem]{Corollary}
\newtheorem{definition}{Definition}

\newcommand{\hmax}{h_{\max}}
\newcommand{\Onotation}[1]{\mathcal{O}\!\left(#1\right)}

\begin{document}

%% No manual line breaks in the title (ACM template requirement); the
%% optional argument provides the short title for the page headers.
\title[Scalable Exact HAC via Sparse Geographic Distance Graphs]{Scalable Exact Hierarchical Agglomerative Clustering via Sparse Geographic Distance Graphs}

\author{Victor Maus}
\correspondingauthor
\orcid{https://orcid.org/0000-0002-7385-4723}
\affiliation{%
  \institution{University of Eastern Finland}
  \city{Kuopio}
  \country{Finland}
}
\affiliation{%
  \institution{WU Vienna University of Economics and Business}
  \city{Vienna}
  \country{Austria}
}
\email{victor.maus@uef.fi}

\author{Vinicius Pozzobon Borin}
\orcid{https://orcid.org/0000-0003-3694-4671}
\affiliation{%
  \institution{University of Oulu}
  \city{Oulu}
  \country{Finland}
}
\email{vinicius.borin@oulu.fi}

%% Concise author list for the page headers.
\renewcommand{\shortauthors}{Maus and Borin}

% -----------------------------------------------------------------------
\begin{abstract}
Exact hierarchical agglomerative clustering (HAC) is limited by the \(\Onotation{n^2}\) time and memory of the full pairwise distance matrix. We present \emph{GSHAC} (Geographically Sparse Hierarchical Agglomerative Clustering), a method that makes exact HAC of millions of geographic features feasible on a commodity workstation. GSHAC replaces the distance matrix with a sparse geographic distance graph that holds only the pairs within a user-specified threshold~\(\hmax\). We build the graph using spatial indexing (KD-tree or ball tree) in \(\Onotation{n \log n + n\bar{k}}\) time, where~\(\bar{k}\) is the mean number of neighbours within~\(\hmax\), and its connected components form independent HAC subproblems. We prove the cluster assignments are \emph{exact}, identical to the dense result, at any cut height \(h \le \hmax\) for single, complete, and average linkage. For single linkage, a minimum-spanning-tree path forms no dense sub-matrix and keeps per-component memory linear in the component's features and edges, independent of its size. When local density is bounded, as it usually is for geographic data, \(\bar{k}\) stays small as~\(n\) grows and time and memory scale near-linearly. On a 2-million-point sample of GeoNames (global named places), GSHAC clusters at 1, 5, 10, and 20\,km in under 2.5\,minutes each on one 30\,GiB workstation, with whole-pipeline peak memory between 0.5 and about 22\,GiB, where the dense matrix would need \(\approx\!29\,\text{TiB}\). To our knowledge, these are the largest demonstrations of exact HAC on geographic data on a single workstation. The implementation follows the scikit-learn interface for use in existing GIS workflows.
\end{abstract}

% CCS concepts
\begin{CCSXML}
<ccs2012>

   <concept>
       <concept_id>10002951.10003227.10003351.10003444</concept_id>
       <concept_desc>Information systems~Clustering</concept_desc>
       <concept_significance>500</concept_significance>
    </concept>

   <concept>
       <concept_id>10002951.10003227.10003236.10003237</concept_id>
       <concept_desc>Information systems~Geographic information systems</concept_desc>
       <concept_significance>500</concept_significance>
    </concept>

   <concept>
       <concept_id>10003752.10003809</concept_id>
       <concept_desc>Theory of computation~Design and analysis of algorithms</concept_desc>
       <concept_significance>300</concept_significance>
       </concept>

   <concept>
       <concept_id>10010147.10010257</concept_id>
       <concept_desc>Computing methodologies~Machine learning</concept_desc>
       <concept_significance>300</concept_significance>
    </concept>

 </ccs2012>
\end{CCSXML}

\ccsdesc[500]{Information systems~Clustering}
\ccsdesc[500]{Information systems~Geographic information systems}
\ccsdesc[300]{Theory of computation~Design and analysis of algorithms}
\ccsdesc[300]{Computing methodologies~Machine learning}

\keywords{sparse distance graph, spatial index, scalable algorithm, fastcluster, scikit-learn}

\maketitle

\section{Introduction}
\label{sec:intro}

Hierarchical agglomerative clustering (HAC) captures multi-scale data structures through a dendrogram that can be cut at varying levels, eliminating the need to predefine the number of clusters~\cite{mullner2013fastcluster}. This makes HAC highly effective for spatial and geographic analysis. Since spatial boundaries rely on acceptable thresholds of internal variation, the dendrogram naturally accommodates the nested, scale-dependent nature of geographic data. As a result, HAC is extensively applied for exploratory analysis, anomaly detection, and delineating regional hierarchies~\cite{sadeghi2025clustering}. Unlike partitional methods such as $k$-means, it does not assume convex cluster shapes or require a priori parameterisation of cluster counts.

Despite these advantages, modern spatial datasets often include millions of observations, such as global land-use inventories, remote sensing segmentations, biodiversity occurrence records, and geographic gazetteers. Exact HAC cannot scale to datasets of this magnitude. Dendrogram construction inherently relies on a full pairwise distance matrix, demanding $\Onotation{n^2}$ memory and $\Onotation{n^2}$ to $\Onotation{n^2 \log n}$ time complexity~\cite{mullner2013fastcluster, nguyen2014sparsehc, gagolewski2016genie}. Consequently, exact multi-scale analysis becomes infeasible on commodity hardware, forcing practitioners to compromise through data downsampling, heuristic approximations, or rigid partitional alternatives.

Existing scalable approaches do not close this gap for \emph{exact} geographic HAC (Section~\ref{sec:related}). Methods that summarise or approximate the hierarchy give up exactness. Distributed single-linkage algorithms keep it but need many-core or cluster hardware, and the memory-efficient exact methods that exploit a distance cutoff expect the sparse matrix as input. For geographic data specifically, Maus (2026) recently proposed a graph-indexing method to compute sparse pairwise distances for HAC suitable for large spatial datasets \cite{maus2026data}. By applying a maximum spatial distance threshold, the method partitions geographic features into independent groups that can be computed efficiently with low memory requirements \cite{maus2026data}. However, this previous work focused strictly on mining land use and does not provide formal mathematical proofs that the approach generalises across standard linkage methods, nor does it prove that it strictly reduces computing time and memory bounds.

This paper presents GSHAC (Geographically Sparse Hierarchical Agglomerative Clustering), which makes exact HAC of millions of geographic features feasible on a single workstation. It relies on Tobler's first law of geography, that near things are more related than distant things~\cite{tobler1970computer}. For clustering, this means that features far apart in space merge only in the last stages of the hierarchy. GSHAC therefore never forms the dense distance matrix. It builds a sparse geographic distance graph by spatial indexing (KD-trees~\cite{bentley1975quad} or ball trees~\cite{omohundro1989balltree}) in \(\Onotation{n \log n + n\bar{k}}\) time, where \(\bar{k}\) is the mean number of neighbours within a user-specified threshold \(\hmax\). When the local density is bounded (\(\bar{k} = \Onotation{1}\) as \(n\) grows), as is typical of geographic data, this reduces to \(\Onotation{n \log n}\) time and \(\Onotation{n\bar{k}} = \Onotation{n}\) memory, near-linear in \(n\). Below \(\hmax\), the connected components of the graph are independent subproblems. One \(\Onotation{n^2}\) problem becomes many small exact ones, and the cluster assignments are identical to dense HAC at every cut \(h \le \hmax\). Our contributions are:

\begin{enumerate}

  \item A pipeline for exact HAC of large geographic datasets that combines spatial indexing, connected-component decomposition, and single linkage computed from the minimum spanning tree (MST). We prove that it is exact at any cut height \(h \le \hmax\) for single, complete, and average linkage, and more generally for any linkage whose inter-cluster dissimilarity is bounded below by the minimum inter-point distance and whose merge heights never decrease. Ward linkage is exact in practice for spatially compact clusters, but this is not guaranteed in general.

  \item An evaluation on a synthetic fixed-domain benchmark (\(n\) up to \(50{,}000\)), on the standard 2-D two-moons and two-circles benchmarks to illustrate exactness, and on a 2-million-point sample of place names from GeoNames (\(n = 2{,}000{,}000\))~\cite{geonames2025}. We compare against \texttt{fastcluster}, PySAL, scikit-learn, and HDBSCAN.

  \item An open-source Python implementation with a scikit-learn-compatible \texttt{SparseAgglomerativeClustering} estimator, available at \url{https://github.com/mine-the-gap/gshac}.

\end{enumerate}

\section{Related work}
\label{sec:related}

Standard HAC libraries like \texttt{fastcluster}~\cite{mullner2013fastcluster} offer optimised routines for common linkage criteria, but they still depend on the $\Onotation{n^2}$ memory footprint of a dense dissimilarity matrix. Single linkage is the notable exception; because it maps directly to the minimum spanning tree (MST) of the distance graph~\cite{gower1969minimum}, SLINK can compute it in $\Onotation{n^2}$ time and linear space~\cite{sibson1973slink}. Parallel frameworks push this boundary further. Fast Euclidean-MST algorithms scale single linkage to tens of millions of points~\cite{wang2021fast}, distributed variants offload the MST across MapReduce or MPI clusters~\cite{jin2015incremental, hendrix2013scalable}, and ParChain~\cite{yu2021parchain} handles exact complete, average, and Ward linkage on metric data via parallel nearest-neighbour-chain merges. However, these tools require multi-socket or distributed hardware with hundreds of gigabytes of RAM. They evaluate distances dynamically across the full metric space rather than restricting computation to connected components of a thresholded graph. SparseHC~\cite{nguyen2014sparsehc} preserves exactness by streaming a pre-computed distance-cutoff matrix chunk by chunk, but it still expects this sparse matrix to be materialised and pre-sorted on disk. In spatial data science, ClustGeo~\cite{chavent2018clustgeo} incorporates spatial contiguity into Ward clustering, yet it retains an $\Onotation{n^2}$ dense matrix and therefore remains confined to small regional samples. In contrast, GSHAC constructs the thresholded graph directly from spatial coordinates on a single workstation.

%\paragraph{Scalable but approximate hierarchical clustering.}
Relaxing exactness opens up several approximate routes. Classic techniques build intermediate summaries over a few linear scans, such as BIRCH's clustering features~\cite{zhang1996birch} and CURE's representative points~\cite{guha1998cure}, or use locality-sensitive hashing to bypass quadratic runtimes~\cite{moseley2020hierarchical}. At billion-point scales, Sub-Cluster Component (SCC) merges independent clusters~\cite{monath2021scalable}, while Grinch dynamically updates trees through local rotations and global grafts~\cite{monath2019scalable}. Both yield non-binary or heuristic trees that only match exact HAC under narrow conditions~\cite{monath2021scalable}. A related thread runs hierarchical clustering directly on a provided sparse similarity graph. Notable examples include $(1+\epsilon)$-approximate average linkage in near-linear time~\cite{dhulipala2021hierarchical}, its poly-logarithmic depth parallelization~\cite{dhulipala2022hierarchical}, TeraHAC for trillion-edge graphs~\cite{dhulipala2023terahac}, and approximate centroid linkage via dynamic nearest-neighbour search~\cite{bateni2024centroid}. Approximations are necessary here because, under standard fine-grained complexity assumptions, exact average-linkage HAC on an arbitrary sparse graph admits no near-linear combinatorial algorithm and resists tree parallelisation~\cite{bateni2024hard}. However, while these methods take an abstract graph as input and compromise exactness to scale, GSHAC uses the sparse spatial graph purely to divide coordinates into independent components. For average linkage, it evaluates the exact distance submatrix within each component, so the graph partitions the problem rather than redefining the metric. This preserves exactness and avoids the hardness lower bound. Centroid linkage and Ward fall outside the guarantees in Section~\ref{sec:correctness} for distinct reasons. Neither merge cost is lower-bounded by the minimum pairwise point distance, and centroid linkage can produce inversions where merge heights decrease~\cite{mullner2011modern}. Other approaches like Genie~\cite{gagolewski2016genie} modify the objective directly, using a Gini index to curb chaining.

%\paragraph{Spatial clustering tools.}
Density-based algorithms like HDBSCAN~\cite{mcinnes2017hdbscan} and OPTICS~\cite{ankerst1999optics} naturally leverage spatial indexing and scale well, but they find high-density regions rather than solving distance-threshold HAC. In GIS, hierarchical methods are widely recognised as critical for multi-scale analysis yet remain bottlenecked by their quadratic overhead~\cite{sadeghi2025clustering}. Practical Python toolkits leave this unaddressed. For example, scikit-learn's \texttt{AgglomerativeClustering} supports a sparse \texttt{connectivity} matrix~\cite{pedregosa2011sklearn}, but it still computes all pairwise coordinates internally and lacks component-level decomposition for large $n$. PySAL's \texttt{libpysal.weights.DistanceBand}~\cite{rey2007pysal} builds distance graphs via KD-trees, yet it discards the actual metric distances in favour of binary or inverse weights and provides no hierarchical clustering. PAGE~\cite{alrashid2023page} parallelises spatially constrained regionalisation to output a single flat partition under regional size bounds. HierGP~\cite{oje2025hiergp} relies on multi-level grid partitioning with user-tunable accuracy, trading raw point fidelity for speed. Bicriteria Shapes~\cite{rottmann2025bicriteria} computes nested polygon aggregations via graph cuts under perimeter and area constraints. While these tools support spatial multi-scale analysis, they optimise alternative objective functions or bin the underlying points. GSHAC retains every individual feature and outputs an exact dendrogram valid for any threshold cut $h \le \hmax$.

\section{The GSHAC method}
\label{sec:theoretical}

\subsection{Problem}
\label{sec:problem}

Let \(\mathcal{X} = \{x_1, \ldots, x_n\}\) be a set of~\(n\) spatial features with pairwise distance function~\(d: \mathcal{X} \times \mathcal{X} \to \mathbb{R}_{\ge 0}\) (geodesic or projected Euclidean). Let \(\hmax \in \mathbb{R}_{>0}\) be a user-specified distance threshold, and \(H \subseteq (0, \hmax]\) be the finite set of cut heights at which cluster assignments are required.

\begin{definition}[Geographic distance graph]
Let $\mathcal{V} = \{1, \ldots, n\}$ index the features in $\mathcal{X}$. The geographic distance graph at threshold $\hmax$ is the weighted undirected graph $G_{\hmax} = (\mathcal{V}, E, w)$, where:
\begin{itemize}
    \item $\mathcal{V}$ is the set of nodes, one per spatial feature;
    \item $E = \{\{i,j\} : 1 \le i < j \le n,\; d(x_i, x_j) \le \hmax\}$ is the set of undirected edges, connecting every pair of features whose distance does not exceed $\hmax$; and
    \item $w : E \to \mathbb{R}_{\ge 0}$, defined by $w(\{i,j\}) = d(x_i, x_j)$, is the edge weight function, assigning to each edge its exact pairwise distance.
\end{itemize}
Pairs $\{i,j\}$ with $d(x_i, x_j) > \hmax$ are absent from $E$ and contribute no entry to the distance graph, yielding the sparsity exploited by GSHAC.
\end{definition}

\begin{definition}[Connected component]
A connected component of~\(G_{\hmax}\) is a maximal subset of node indices \(\mathcal{C} \subseteq \mathcal{V}\) such that every pair in~\(\mathcal{C}\) is connected by a path of edges in~\(E\). We write \(\mathcal{C}_1, \ldots, \mathcal{C}_K\) for the \(K\) connected components of~\(G_{\hmax}\), which partition~\(\mathcal{V}\).
\end{definition}

The goal is to compute, for each \(h \in H\), the same cluster assignment that would result from running HAC on the full dense distance matrix \(\mathbf{D} = [d(x_i, x_j)]_{i,j=1}^{n}\), while storing only the entries of~\(\mathbf{D}\) that correspond to edges of \(G_{\hmax}\), together with at most one component's sub-matrix at a time.

\subsection{Algorithm}
\label{sec:algorithm}

Algorithm~\ref{alg:main} groups the work into two phases. \emph{Phase~1} builds the sparse geographic distance graph: (i)~spatial index construction, (ii)~candidate pair discovery via range queries and exact distance computation, and (iii)~connected-component decomposition. \emph{Phase~2} runs HAC independently within each component, with a fast MST path for single linkage that never materialises a dense sub-matrix. Unlike the dense baseline, GSHAC computes and stores only the \(m = |E|\) edge distances of~\(G_{\hmax}\), as the sparse symmetric matrix~\(\mathbf{S}\) of its edge weights. With \(\bar{k} = 2m/n\), the mean node degree of~\(G_{\hmax}\), i.e.\ the mean neighbourhood size (mean number of neighbours within~\(\hmax\)), this is \(m = n\bar{k}/2 = \Onotation{n\bar{k}}\) distances.

\begin{algorithm}[!ht]
\caption{GSHAC: exact HAC via a sparse geographic distance graph}
\label{alg:main}
\begin{algorithmic}[1]
  \Require Point coordinates $\mathcal{X} = \{x_1,\ldots,x_n\}$;
    distance threshold $\hmax$;
    linkage method $\ell \in \{\text{single},\allowbreak\ \text{complete},\allowbreak\ \text{average},\allowbreak\ \text{Ward}\}$;
    cut heights $H \subseteq (0, \hmax]$
  \Ensure Cluster-label vector $\mathbf{c}^{(h)} \in \mathbb{Z}^n$ for each $h \in H$

  \Statex \textbf{Phase 1: Build sparse geographic distance graph}
  \State Build spatial index $\mathcal{T}$ over $\mathcal{X}$
    \Comment{Euclidean or haversine}
  \State $\mathbf{S} \leftarrow$ empty $n \times n$ sparse symmetric matrix
  \For{$i = 1$ \textbf{to} $n$}
    \State $N_i \leftarrow \mathcal{T}.\textsc{RangeQuery}(x_i, \hmax)$
      \Comment{indices $j \in \mathcal{V}$ within $\hmax$}
    \For{$j \in N_i$ \textbf{with} $j > i$}
      \Comment{upper triangle}
      \State $s_{ij} \leftarrow d(x_i, x_j)$
      \State $\mathbf{S}[i,j] \leftarrow \mathbf{S}[j,i] \leftarrow s_{ij}$
    \EndFor
  \EndFor
  \State $\{\mathcal{C}_1,\ldots,\mathcal{C}_K\} \leftarrow \textsc{ConnectedComponents}(\mathbf{S})$

  \Statex \textbf{Phase 2: Component-wise HAC}
  \State Initialise $\mathbf{c}^{(h)} \in \mathbb{Z}^n$ to $\mathbf{0}$ for each $h \in H$
  \State $g \leftarrow 0$
    \Comment{global cluster-ID offset}
  \For{$k = 1$ \textbf{to} $K$}
    \If{$|\mathcal{C}_k| = 1$}
      \Comment{singleton: no merges to compute}
      \State $g \leftarrow g + 1$
      \For{$h \in H$}
        \State $\mathbf{c}^{(h)}[\mathcal{C}_k] \leftarrow g$
      \EndFor
      \State \textbf{continue}
    \EndIf
    \State $\mathbf{S}_k \leftarrow$ sparse sub-matrix of $\mathbf{S}$ restricted to $\mathcal{C}_k$
    \If{$\ell = \text{single}$}
      \Comment{MST fast path: no dense sub-matrix}
      \State $\mathcal{M}_k \leftarrow \textsc{MinimumSpanningTree}(\mathbf{S}_k)$
      \State $Z_k \leftarrow \textsc{LinkageFromMST}(\mathcal{M}_k)$
        \Comment{union-find over MST edges}
    \Else
      \Comment{dense path for complete / average / Ward}
      \State $\mathbf{D}_k \leftarrow \bigl[d(x_i, x_j)\bigr]_{i,j \in \mathcal{C}_k}$ recomputed from coordinates
        \Comment{includes pairs $> \hmax$}
      \State $Z_k \leftarrow \textsc{Linkage}(\mathbf{D}_k,\, \ell)$
    \EndIf
    \State $q_k \leftarrow 0$
    \For{$h \in H$}
      \State $\boldsymbol{\lambda}^{(h)} \leftarrow \textsc{CutTree}(Z_k, h)$
        \Comment{local labels $1,\ldots,\max \boldsymbol{\lambda}^{(h)}$}
      \State $\mathbf{c}^{(h)}[\mathcal{C}_k] \leftarrow g + \boldsymbol{\lambda}^{(h)}$
      \State $q_k \leftarrow \max(q_k,\, \max \boldsymbol{\lambda}^{(h)})$
    \EndFor
    \State $g \leftarrow g + q_k$
  \EndFor
  \State \textbf{return} $\{\mathbf{c}^{(h)}\}_{h \in H}$
\end{algorithmic}
\end{algorithm}

\subsection{Correctness}
\label{sec:correctness}

Because all cross-component pairs are farther apart than \(\hmax\) by definition, no merge at height \(h \le \hmax\) can join features from different connected components. This holds whenever the linkage's inter-cluster dissimilarity is bounded below by the minimum inter-point distance, and its merge heights never decrease, as with the most common criteria. For single linkage, the merge height is the minimum inter-cluster distance, and all cross-component minima exceed \(\hmax\). For complete linkage, it is the maximum inter-cluster distance, which is at least the minimum. For average linkage, it is the mean over all cross-cluster pairs, which exceeds \(\hmax\) when every pair does. In all three cases, the dissimilarity between a merged cluster and any other cluster is at least the dissimilarity at which the merge happened, so merge heights never decrease~\cite{milligan1979ultrametric}. Ward linkage meets this second condition but not the first. Its merge criterion is the increase in within-cluster variance, which is proportional to the distance between cluster \emph{centroids} (scaled by cluster sizes) and is \emph{not} bounded below by the minimum inter-point distance. We treat Ward separately after the theorem.

\begin{theorem}[Exact equivalence]
\label{thm:exact}
Let $\ell$ be a linkage method whose inter-cluster dissimilarity is bounded below by the minimum inter-point distance, $\delta_\ell(\mathcal{A}, \mathcal{B}) \ge \min_{i \in \mathcal{A},\, j \in \mathcal{B}} d(x_i, x_j)$, and whose merge heights are non-decreasing along the merge sequence (both conditions hold for single, complete, and average linkage). Assume that all pairwise distances in $\mathcal{X}$ are distinct and, for average linkage, that no two inter-cluster dissimilarities compared during HAC coincide. Then, for any cut height $h \le \hmax$, the cluster assignment produced by Algorithm~\ref{alg:main} is identical (up to a relabelling of clusters) to the assignment produced by HAC on the full dense distance matrix~$\mathbf{D}$.
\end{theorem}

\begin{proof}
We show that the sequence of merges at heights $\le \hmax$ in dense HAC on $\mathbf{D}$ decomposes into independent intra-component merges within each $\mathcal{C}_k$.

\noindent\emph{Inter-component merges.} Let $\mathcal{C}_k$ and $\mathcal{C}_{k'}$ ($k \neq k'$) be two distinct connected components of $G_{\hmax}$, and let $\mathcal{A} \subseteq \mathcal{C}_k$ and $\mathcal{B} \subseteq \mathcal{C}_{k'}$ be any two non-empty sub-clusters contained within them. Because $\mathcal{C}_k$ and $\mathcal{C}_{k'}$ are disconnected in $G_{\hmax}$, every pair $(i, j)$ with $i \in \mathcal{A}$ and $j \in \mathcal{B}$ satisfies $d(x_i, x_j) > \hmax$. By the hypothesis on $\ell$,
\[
  \delta_\ell(\mathcal{A}, \mathcal{B}) \ge \min_{i \in \mathcal{A},\, j \in \mathcal{B}} d(x_i, x_j) \ge \min_{i \in \mathcal{C}_k,\, j \in \mathcal{C}_{k'}} d(x_i, x_j) > \hmax.
\]
Since merge heights are non-decreasing, the merges at height $\le \hmax$ form an initial segment of the merge sequence. Initially, every cluster is a singleton and so lies within one component, and by the bound above a merge at height $\le \hmax$ between two clusters that each lie within one component joins two clusters of the same component; by induction over that initial segment, every cluster formed at height $\le \hmax$ lies within a single connected component.

\noindent\emph{Intra-component merges.} For $\mathcal{A}, \mathcal{B} \subseteq \mathcal{C}_k$, $\delta_\ell(\mathcal{A}, \mathcal{B})$ is a function of the entries of $\mathbf{D}_k = \mathbf{D}[\mathcal{C}_k, \mathcal{C}_k]$ alone, and whenever the globally closest pair of clusters lies within $\mathcal{C}_k$ it is also the closest pair within $\mathcal{C}_k$ (unique, since no two candidate merges tie). Because no cross-component merge occurs at height $\le \hmax$, the merges within $\mathcal{C}_k$ up to height $\hmax$ therefore occur in the same order, and at the same heights, as in HAC on $\mathbf{D}_k$. For complete and average linkage, Algorithm~\ref{alg:main} forms exactly this $\mathbf{D}_k$, so HAC on $\mathbf{D}_k$ executes the same intra-component merges as HAC on $\mathbf{D}$ up to $\hmax$. For single linkage, every pair $\{i,j\} \subseteq \mathcal{C}_k$ with $d(x_i, x_j) > \hmax$ forms a cycle with a path of edges of $G_{\hmax}$ joining $i$ and $j$ (one exists because $\mathcal{C}_k$ is connected), and it is the heaviest edge on that cycle; by the cycle property it belongs to no minimum spanning tree. Every minimum spanning tree of the complete graph on $\mathcal{C}_k$ therefore uses only edges of $\mathbf{S}_k$; since $\mathbf{S}_k$ is connected, its minimum spanning tree $\mathcal{M}_k$ attains the same weight and is itself such a tree (the unique one, as all distances are distinct), and it yields the exact single-linkage dendrogram of $\mathcal{C}_k$~\cite{gower1969minimum}.

Finally, Algorithm~\ref{alg:main} cuts each component's dendrogram at every $h \in H$ and gives each component a disjoint range of labels, so its output at height $h$ is the union of the per-component partitions; by the two parts above, this equals the partition produced by dense HAC at $h$, up to a relabelling of clusters.
\end{proof}

\paragraph{Tied distances.} The distinctness assumptions exclude exactly tied dissimilarities. With ties, the complete- and average-linkage dendrogram is not unique, and \emph{any} exact algorithm is defined only up to tie-breaking; for continuous coordinates, ties have probability zero. Single linkage needs no such assumption, because the partition at any cut \(h \le \hmax\) is determined by the set of pairwise distances not exceeding~\(h\), whatever order tied merges are processed in (see Section~\ref{sec:exactness}). Co-located features, which tie at distance~0, do no harm for any of the three linkages. Each group of co-located features merges at height~0, before any merge at positive height and into the same cluster whatever the order, and from then on behaves as one feature with multiplicity. The partitions at every positive cut height are therefore unaffected. Section~\ref{sec:implementation} describes how the sparse graph stores such pairs, which places a lower limit of one distance unit on single-linkage cut heights.

\begin{corollary}[Component decomposition]
Exact HAC at every cut height $h \le \hmax$ decomposes over the connected components of $G_{\hmax}$ for single, complete, and average linkage. For single linkage, the sparse matrix $\mathbf{S}$ alone determines the dendrogram below $\hmax$; for complete and average linkage, each component's distances beyond $\hmax$ are recomputed from coordinates.
\end{corollary}

\paragraph{Ward linkage.} Ward's merge criterion is the increase in within-cluster sum of squares. It is proportional to the distance between cluster centroids (scaled by cluster sizes) and is \emph{not} bounded below by the minimum inter-point distance. One component can surround another. Take a ring of points whose nearest-neighbour spacing is \(\le \hmax\) but whose radius exceeds \(\hmax\), with a central point farther than \(\hmax\) from every ring point. The two components then have nearly coincident centroids even though every cross-component pair exceeds \(\hmax\). Dense Ward may then merge them below \(\hmax\), which the component decomposition cannot reproduce. Theorem~\ref{thm:exact} therefore does not extend to Ward in general. This configuration does not arise for the spatially compact clusters typical of geographic data. Because GSHAC recomputes each component's full sub-matrix from coordinates, Ward is exact \emph{within} every component, and it reproduced the dense Ward partition in all of our experiments (Section~\ref{sec:exactness}).

\subsection{Complexity analysis}
\label{sec:complexity}

Recall that $m = |E|$ and that $\bar{k} = 2m/n$ is the mean neighbourhood size. For each component $\mathcal{C}_k$, let $c_k = |\mathcal{C}_k|$ be its size and $m_k$ its number of edges, so that $\sum_k c_k = n$ and $\sum_k m_k = m$. Table~\ref{tab:complexity} summarises the complexity of each step.

\begin{table}[ht]
\caption{Complexity of Algorithm~\ref{alg:main} vs.\ dense baseline. 
Here $\bar{k}$ is the mean neighbourhood size, $c_k$ the size of component~$k$, and $m_k$ its number of edges. Dense-path HAC times assume \texttt{fastcluster}; priority-queue-based implementations add a $\log$ factor.}
\label{tab:complexity}
\centering
\small
\begin{tabular}{lccc}
\toprule
& \multicolumn{2}{c}{\textbf{Proposed (GSHAC)}} & \\
\cmidrule(lr){2-3}
\textbf{Step} & \textbf{Single (MST)} & \textbf{Complete / Avg / Ward} & \textbf{Dense baseline} \\
\midrule
Spatial index build       & $\Onotation{n \log n}$ & $\Onotation{n \log n}$ & --- \\
Range queries             & $\Onotation{n \log n + m}$ & $\Onotation{n \log n + m}$ & --- \\
Distance computations     & $\Onotation{m}$ & $\Onotation{m}$ & $\Onotation{n^2}$ \\
Connected components      & $\Onotation{n + m}$ & $\Onotation{n + m}$ & --- \\
HAC (all components)      & $\Onotation{\sum_k m_k \log m_k}$ & $\Onotation{\sum_k c_k^2}$ & $\Onotation{n^2}$ \\
Cut + label assignment    & $\Onotation{|H| \cdot n}$ & $\Onotation{|H| \cdot n}$ & $\Onotation{|H| \cdot n}$ \\
\midrule
Sparse graph storage      & $\Onotation{m}$ & $\Onotation{m}$ & --- \\
HAC working memory        & $\Onotation{\max_k (c_k + m_k)}$ & $\Onotation{\max_k c_k^2}$ & $\Onotation{n^2}$ \\
\midrule
\textbf{Total time}       & $\Onotation{n \log n + \sum_k m_k \log m_k}$ & $\Onotation{n \log n + \sum_k c_k^2}$ & $\Onotation{n^2}$ \\
\textbf{Total memory}     & $\Onotation{n \bar{k} + \max_k (c_k + m_k)}$ & $\Onotation{n \bar{k} + \max_k c_k^2}$ & $\Onotation{n^2}$ \\
\bottomrule
\end{tabular}
\end{table}

The HAC step within component~\(k\) costs \(\Onotation{c_k^2 \log c_k}\) with priority-queue-based implementations, or \(\Onotation{c_k^2}\) with \texttt{fastcluster}~\cite{mullner2013fastcluster}. For single linkage, since \(\sum_k m_k = m = n\bar{k}/2\) and \(\log m_k \le \log m = \Onotation{\log n}\), the total time is \(\Onotation{(n+m)\log n} = \Onotation{n(1+\bar{k})\log n}\) (\(\Onotation{n\bar{k}\log n}\) for \(\bar{k} \ge 1\)), which reduces to \(\Onotation{n\log n}\) when \(\bar{k}\) is bounded. On the dense path, peak memory is \(\Onotation{n\bar{k} + \max_k\, c_k^2}\), where the second term reflects the dense sub-matrix \(\mathbf{D}_k\) materialised for the largest component; 
components are processed sequentially and freed after each HAC call. The memory advantage over the \(\Onotation{n^2}\) dense baseline therefore holds when no single component dominates, i.e.\ when \(\max_k\, c_k \ll n\), which is the expected regime under the geographic sparsity assumption.

\paragraph{Practical regime.}
Geographically clustered data, the common case for natural and anthropogenic features, give many small components (\(c_k \ll n\)). Then \(\sum_k c_k^2 \ll n^2\), and GSHAC is much faster than the dense baseline. The worst case is a graph in which nearly all pairs lie within \(\hmax\). GSHAC then does slightly more work than the dense baseline. Sorting \(\Onotation{n^2}\) edges for the minimum spanning tree adds a \(\log n\) factor to the time of single linkage, and sparse storage costs more per pair than a condensed array. The reduction in distance computations is exactly \(\binom{n}{2}/m = (n-1)/\bar{k}\), linear in \(n/\bar{k}\).

\paragraph{Two scaling regimes.} The mean neighbourhood size \(\bar{k}\) is not a fixed property of the method. It depends on how a dataset grows, and two regimes bound the behaviour. If the spatial domain is held fixed while \(n\) increases, point density rises and \(\bar{k}\) grows in proportion to \(n\). The graph then holds \(\Onotation{n^2}\) edges, and GSHAC and the dense baseline both scale quadratically (single linkage up to a \(\log n\) factor). GSHAC stays cheaper while the graph holds a small fraction of all pairs, but loses its advantage as that fraction grows (Section~\ref{sec:scaling}). If instead the domain grows with \(n\) at roughly constant density, as geographic inventories often do when larger datasets cover larger areas, \(\bar{k}\) stays bounded, and the graph holds \(\Onotation{n}\) edges. The \(\Onotation{n\bar{k}\log n}\) time and \(\Onotation{n\bar{k}}\) memory are then near-linear in \(n\). In both regimes, the threshold \(\hmax\) bounds what is stored, so as density grows, the method degrades gradually towards the dense cost instead of failing abruptly. Sections~\ref{sec:scaling} and~\ref{sec:realworld} evaluate both regimes.

\paragraph{MST path for single linkage.} The single-linkage dendrogram is the MST of the complete graph~\cite{gower1969minimum}. Computing the MST of each component's \emph{sparse} subgraph gives the exact single-linkage dendrogram in \(\Onotation{m_k \log m_k}\) time and \(\Onotation{c_k + m_k}\) memory per component, and no dense sub-matrix is ever formed. Each component's MST has exactly \(c_k - 1\) edges (\(n - K\) in total), whereas GSHAC computes all \(m = \Onotation{n\bar{k}}\) edges of \(G_{\hmax}\). Choosing the smallest \(\hmax\) the application allows therefore minimises \(\bar{k}\), and with it the computational cost.

\subsection{Linkage methods and scalability}
\label{sec:linkage-scaling}

Theorem~\ref{thm:exact} covers single, complete, and average linkage, but their resource profiles differ, which determines which datasets each can handle. Single linkage is obtained from the minimum spanning tree of each component's sparse subgraph (Section~\ref{sec:algorithm}). It uses \(\Onotation{c_k + m_k}\) memory and forms no dense sub-matrix. Its cost depends on the number of stored edges, and one component may span a large fraction of \(n\) without exhausting memory. Complete, average, and Ward linkage have no such path. Their merge heights depend on intra-component pairs farther apart than \(\hmax\), which the threshold graph omits, so each component's full \(c_k \times c_k\) sub-matrix is recomputed from coordinates and passed to a dense routine. Peak memory is then \(\Onotation{\max_k c_k^2}\), set by the largest component. The sparse graph can be small while one large component already needs a dense sub-matrix that does not fit in memory. These linkages are feasible only when the largest component is small enough to materialise. For this reason, the scaling evaluation (Section~\ref{sec:scaling}) uses single linkage. Section~\ref{sec:otherlinkage} returns to complete, average, and Ward linkage and how to choose among them, using the component-size distributions of the evaluation datasets to make the limit concrete.

\subsection{Implementation}
\label{sec:implementation}

The pipeline runs on the scientific Python stack with its performance-critical paths in C. Range queries use \texttt{scipy.\allowbreak spatial.\allowbreak cKDTree} for projected coordinates and \texttt{sklearn.\allowbreak neighbors.\allowbreak BallTree} with the haversine metric for geographic ones, and the sparse graph and its connected components use \texttt{scipy.sparse}. Because sparse-matrix routines can discard explicitly stored zeros, co-located features (distance~0) are stored in the graph with weight~1 in the distance unit (1\,m when distances are in metres). When co-located features are present, single-linkage partitions therefore match dense HAC at every cut height \(h\) with \(1 \le h \le \hmax\) (in the chosen distance unit). In contrast, the dense path for complete and average linkage recomputes the true zero distances. In our benchmarks, co-located features occur only in GeoNames, whose distances are in metres, so the 1\,m floor lies far below every cut height evaluated there (\(h \ge 1\)\,km). Co-located features also create tied distances, which the remark on tied distances in Section~\ref{sec:correctness} shows to be harmless. A dedicated C extension (\texttt{\_gshac.so}) implements union-find construction of the linkage matrix from MST edges, the haversine distance kernel, and batch dendrogram cutting at several heights, releasing Python's global interpreter lock (GIL); a pure-Python fallback is used when it is not built. Only component iteration and label assembly run in Python.

The two linkage paths of Section~\ref{sec:algorithm} map onto two code paths. Single linkage takes the minimum spanning tree of each component's sparse sub-matrix and builds the linkage matrix by a union-find pass over its edges, forming no dense matrix. Complete, average, and Ward recompute each component's full sub-matrix from coordinates and pass it to \texttt{fastcluster}~\cite{mullner2013fastcluster}, with \texttt{scipy.cluster.hierarchy} as a fallback. The implementation is exposed through a scikit-learn-compatible estimator, \texttt{SparseAgglomerativeClustering}, and a \texttt{geographic\_connectivity} helper that returns a connectivity matrix for scikit-learn's own \texttt{AgglomerativeClustering}, so that adopting GSHAC in an existing GIS or clustering pipeline means changing a class name. Appendix~\ref {sec:system} details distance dispatch by coordinate reference system (CRS), support for line and polygon geometries, and the choice of the haversine kernel for the benchmarks.

\section{Experimental design}
\label{sec:expdesign}

The implementation is available at \url{https://github.com/mine-the-gap/gshac} under the GPL-3.0 licence. We evaluate GSHAC along four research questions:

\begin{description}[leftmargin=0pt,topsep=2pt,itemsep=1pt]\sloppy
  \item[RQ1 (Exactness):] Does GSHAC produce identical cluster assignments to the dense baseline across the supported linkage methods and configurations?
  \item[RQ2 (Efficiency where dense HAC is feasible):] Where the dense baseline can run, how much does GSHAC reduce time and memory, and how does the advantage vary with graph density?
  \item[RQ3 (Scalability where dense HAC is infeasible):] Beyond the dense-feasible range, does GSHAC scale near-linearly under constant density and remain feasible at millions of points, consistent with the predictions of Section~\ref{sec:complexity}?
  \item[RQ4 (Positioning):] How does GSHAC compare with the spatial-clustering tools already in common use (PySAL, scikit-learn's connectivity HAC, and HDBSCAN)?
\end{description}

\paragraph{Datasets.} We use three datasets of increasing size. The first is a synthetic set of clustered point clouds from a Gaussian mixture model, with \(C\) centres placed uniformly in a \(500 \times 500\)\,km projected domain and \(n\) points displaced from them by Gaussian noise~\(\sigma\). Three density scenarios vary \(C\) and~\(\sigma\) relative to~\(\hmax\). The tight scenario (\(C\!=\!100\), \(\sigma\!=\!0.03\,\hmax\)) has compact, well-separated clusters and is the best case. The moderate scenario uses \(C\!=\!50\) and \(\sigma\!=\!0.10\,\hmax\). The loose scenario (\(C\!=\!20\), \(\sigma\!=\!0.30\,\hmax\)) has broad, overlapping clusters and is the adversarial case. Figure~\ref{fig:scenarios} illustrates the three scenarios. We use \(n \in \{1{,}000, 5{,}000, 10{,}000, 25{,}000, 50{,}000\}\) and \(\hmax \in \{10, 20, 50\}\)\,km.

The second dataset is the pair of standard 2-D clustering benchmarks, two moons and two circles (\(n = 500\) each)~\cite{pedregosa2011sklearn}, the usual non-convex shapes for testing density-based clustering. We use them to show visually that GSHAC gives the same partitions as dense single-linkage HAC.

The third is GeoNames (\(n = 2{,}000{,}000\)), a spatially stratified random sample (14.9\%) of the GeoNames \texttt{allCountries} database (\(n_{\text{full}} = 13{,}412{,}837\), CC~BY~4.0)~\cite{geonames2025}, drawn with proportional allocation over \(1^{\circ}\!\times\!1^{\circ}\) grid cells. For GeoNames we use the haversine metric, \(\hmax \in \{1, 5, 10, 20\}\)\,km, and single linkage.

\begin{figure*}[t]
  \centering
  \includegraphics[width=\textwidth]{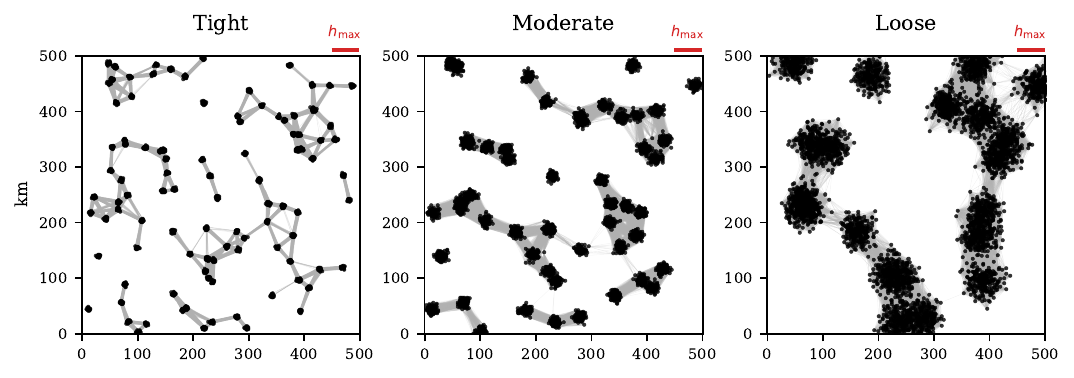}
  \caption{The three synthetic density scenarios at one representative point of the benchmark sweep (\(n = 5{,}000\), \(\hmax = 50\)\,km, \(500\times500\)\,km domain). Each scenario is defined, as in the benchmark, by its number of cluster centres~\(C\) and its Gaussian spread~\(\sigma\). \emph{Tight} (\(C = 100\), \(\sigma = 0.03\,\hmax\)): many compact, well-separated clusters and a sparse graph. \emph{Moderate} (\(C = 50\), \(\sigma = 0.10\,\hmax\)): intermediate overlap. \emph{Loose} (\(C = 20\), \(\sigma = 0.30\,\hmax\)): a few broad clusters that produce a dense, highly connected graph, the adversarial case for GSHAC. Grey edges connect point pairs within~\(\hmax\); the red bar shows the \(\hmax\) scale. The mean neighbourhood size~\(\bar{k}\) grows with \(n\) and \(\hmax\) across the sweep, so no single value characterises a scenario; the benchmarked \(\bar{k}\) values are reported with the scaling results (Section~\ref{sec:scaling}).}
  \Description{Three scatter plots of synthetic point clouds side by side. In the tight scenario, the points form many small, clearly separated blobs; in the moderate scenario, the blobs are larger and begin to overlap; in the loose scenario, the points spread into a few broad overlapping clouds. Grey line segments connect nearby points, becoming progressively denser from the tight to the loose scenario.}
  \label{fig:scenarios}
\end{figure*}

\paragraph{Baselines and comparisons.} (i)~Dense baseline: complete \(\Onotation{n^2}\) matrix via \texttt{scipy\allowbreak .spatial\allowbreak .distance\allowbreak .cdist} for Euclidean (or a manual nested loop for haversine, since \texttt{cdist} lacks a haversine kernel) followed by \texttt{fastcluster.\allowbreak linkage} \cite{mullner2013fastcluster}, the fastest available dense HAC implementation for Python. The dense haversine baseline for the GeoNames dataset is reported analytically because the corresponding matrix exceeds available RAM. (ii)~scikit-learn \texttt{AgglomerativeClustering} with sparse \texttt{connectivity}~\cite{pedregosa2011sklearn} (single linkage only;\footnote{scikit-learn carries a documented bug (issue \#23550~\cite{sklearn_issue23550}) causing incorrect results for complete/average linkage with \texttt{distance\_threshold} and sparse connectivity.}). (iii)~PySAL \texttt{DistanceBand}~\cite{rey2007pysal} (graph construction only). (iv)~\texttt{sklearn.cluster.HDBSCAN} v1.8~\cite{mcinnes2017hdbscan} (scope comparison; non-equivalent outputs).

\paragraph{Metrics and hardware.} We report wall time (seconds) and peak memory (binary units, MiB/GiB) measured with \texttt{tracemalloc}, averaged over \(r = 5\) repetitions. We report peak memory separately for the graph-construction phase, the HAC phase, and the whole pipeline. In the synthetic, constant-\(\bar{k}\), and GeoNames benchmarks, every timed measurement runs in its own fresh subprocess, for each dense and each sparse run at every repetition. For GeoNames, this frees memory between runs. For the dense-vs-sparse comparisons, it makes every measurement \emph{cold}, and both the dense baseline and the sparse pipeline pay the full fresh-process cost. The dense baseline allocates a new \(\Onotation{n^2}\) distance matrix on every call. In a long-lived benchmark process, the memory allocator spreads the cost of that allocation and its first-touch page faults over the repetitions, which understates the wall time a practitioner pays when clustering a dataset \emph{once}. A fresh process per repetition measures that one-shot cost. The one-shot cost in turn depends on how the operating system backs the \(\Onotation{n^2}\) allocation. NumPy requests transparent huge pages for large arrays, and whether the kernel grants them depends on memory fragmentation at run time, which makes the first-touch page-fault cost bimodal. On our hardware the cold dense baseline at \(n = 25{,}000\) takes \({\sim}3.5\)\,s when backed by 2\,MiB huge pages and \({\sim}7.6\)\,s with 4\,KiB pages. The harness therefore turns off the huge-page request (\texttt{NUMPY\_MADVISE\_HUGEPAGE=0}), which fixes every measurement to the deterministic 4\,KiB-page state. We chose 4\,KiB pages because it is the only setting we can reproduce. Huge pages can be requested from user space but not enforced, since the kernel's decision depends on fragmentation outside the harness's control. A long-running analysis machine ends up in this state anyway as its memory fragments, which makes the measurements representative of practice. The setting applies to both methods, although it matters only for the dense baseline's \(\Onotation{n^2}\) allocation. On a system that does grant huge pages, the dense baseline runs up to \({\sim}2\times\) faster, and the wall-time speedups we report shrink accordingly. For the dense-feasible configurations, exactness, measured by the adjusted Rand index (ARI), is checked in a separate untimed in-process pass so that it never disturbs a timed region. Each timed repetition records a fingerprint of its partition after measurement stops, which lets us check that the repetitions agree. The PySAL graph build and the exactness illustration are deterministic single runs, and the scikit-learn end-to-end comparison is timed in-process with both sides measured the same way.ARI and cluster-count agreement assess correctness at each cut height. All experiments ran on an AMD Ryzen~7 PRO 8840U laptop (8 cores, 3.3--5.1\,GHz) with 30\,GiB DDR5 under Ubuntu Linux (kernel 6.17, Python 3.12.3). GSHAC and the dense baseline both run single-threaded. None of the timed kernels (\texttt{cKDTree}, \texttt{BallTree}, \texttt{fastcluster}, the C extension) is multithreaded, and the harness pins all BLAS/OpenMP thread pools to one thread, so every reported speedup compares one core with one core.

%=======================================================================
\section{Results}
\label{sec:results}

We first check exactness, that is, whether GSHAC computes the same partition as a dense baseline (Section~\ref{sec:exactness}). We then measure efficiency where the dense baseline is feasible (Section~\ref{sec:scaling}) and scaling where it is not (Section~\ref{sec:realworld}), and finally compare GSHAC with spatial-clustering tools in common use (Section~\ref{sec:positioning}). 

\subsection{Exactness (RQ1)}
\label{sec:exactness}

Theorem~\ref{thm:exact} guarantees that GSHAC and a dense baseline give the same partition at every cut \(h \le \hmax\) for single, complete, and average linkage. The experiments in this section check that the implementation delivers this guarantee, and that it does so deterministically. For all 36 configurations in which the dense baseline was feasible (\(n \leq 25{,}000\); 3~scenarios \(\times\) 4~sizes \(\times\) 3~\(\hmax\) values), GSHAC with single linkage reproduced the dense baseline's partition at every cut height (identical cluster assignments, ARI\,=\,1.0) in all five repetitions, as Theorem~\ref{thm:exact} predicts. Each timed repetition fingerprints its partition (a hash of its label vectors, computed after measurement stops) inside its own subprocess. The fingerprints agree across all repetitions of every configuration, and a separate untimed pass checks the common partition against the dense baseline. For the 9 sparse-only configurations at \(n = 50{,}000\), the per-repetition fingerprints also agree, and two further independent runs on identical input give byte-identical cluster assignments, which confirms that the results are deterministic. We tested the multi-linkage guarantee separately. On 48 dense-feasible configurations covering single, complete, average, and Ward linkage, GSHAC reproduced the dense \texttt{fastcluster} partition (ARI\,=\,1.0) at every cut. This confirms Theorem~\ref{thm:exact} for single, complete, and average linkage. For Ward, which the theorem does not cover, GSHAC matched the dense partition in every case we tested (cf.\ the Ward remark in Section~\ref{sec:correctness}).

Figure~\ref{fig:exactness} illustrates exactness on two standard non-convex benchmarks, two moons and two circles (\(n = 500\) each). With \(\hmax = 0.20\), GSHAC computes only 5.3\% of the pairwise distances for two moons (6{,}617 of 124{,}750 pairs) and 4.3\% for two circles (5{,}329 of 124{,}750), decomposing each shape into 2~connected components. At every tested cut height \(h \leq \hmax\) the GSHAC partition is identical to that of dense \texttt{fastcluster} (ARI\,=\,1.0): for two moons, from 7~clusters at \(h = 0.10\) down to the 2~true clusters at \(h = 0.20\); for two circles, from 11~clusters at \(h = 0.10\) down to 2. Single linkage recovers the non-convex shapes exactly while computing only a small fraction of the distances.

GSHAC and \texttt{fastcluster} may process tied distances in different orders. GSHAC extracts a minimum spanning tree from the sparse graph and merges via union-find, whereas \texttt{fastcluster} runs a nearest-neighbour-chain algorithm on the full condensed distance vector. The partition at any cut \(h \leq \hmax\) does not depend on merge order, because the set of pairwise distances below the threshold determines it. This is the tie invariance noted after Theorem~\ref{thm:exact} (Section~\ref{sec:correctness}), and it is why the GSHAC and dense columns of Figure~\ref{fig:exactness} are identical although the internal merge sequences differ.

\begin{figure}[!ht]
  \centering
  \includegraphics[width=\columnwidth]{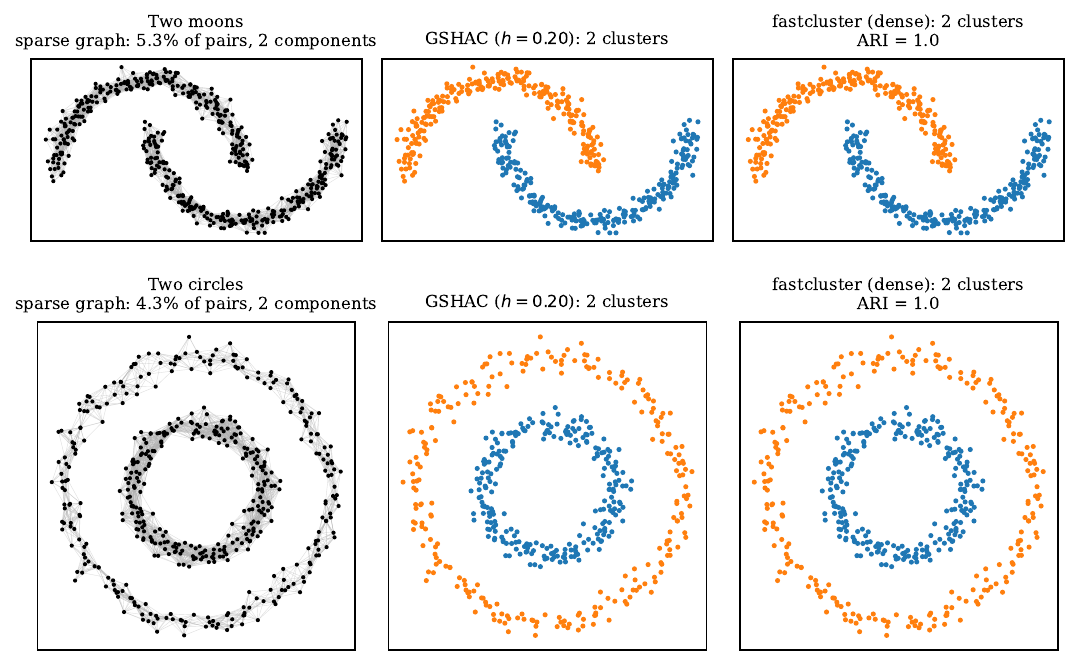}
  \caption{Exactness on two standard non-convex benchmarks (\(n = 500\) each, single linkage). \emph{Rows}: two moons (top), two circles (bottom). \emph{Columns}: (left) the sparse geographic graph, in which grey edges connect point pairs within \(\hmax = 0.20\), so only 5.3\% (moons) and 4.3\% (circles) of all pairs are computed and each shape forms 2~connected components; (middle) GSHAC cluster assignments at \(h = 0.20\); (right) dense \texttt{fastcluster} assignments at \(h = 0.20\). GSHAC recovers the 2~true shapes and matches the dense baseline exactly (ARI\,=\,1.0). Partitions are identical at every tested cut, \(h \in \{0.10, 0.15, 0.20\}\).}
  \Description{A three-by-two grid of scatter plots. The top row shows the two-moons dataset, two interleaving crescent shapes; the bottom row shows the two-circles dataset, one ring inside another. In the left column, grey edges trace each shape's outline, connecting only nearby points. In the middle and right columns, points are coloured by cluster membership, and the colouring is visually identical between the middle (GSHAC) and right (dense baseline) columns, with each shape in a distinct colour.}
  \label{fig:exactness}
\end{figure}

\subsection{Efficiency where dense HAC is feasible (RQ2)}
\label{sec:scaling}

The fixed-domain synthetic benchmark is the only setting in which the dense baseline runs to completion. It is where we compare GSHAC with the dense baseline directly. With the domain fixed, \(\bar{k} \propto n\) (Section~\ref{sec:complexity}). In this regime both methods scale quadratically and GSHAC's wall-time advantage is smallest, while its memory advantage, which comes from storing only the sparse graph, holds throughout. All runs here use single linkage.

Figure~\ref{fig:scaling} shows wall time and peak memory as functions of~\(n\) for both approaches, for the three scenarios and three~\(\hmax\) values. Because point density increases with~\(n\) in this fixed-domain setup and \(\bar{k}\) grows with it, both approaches scale quadratically (Section~\ref{sec:constk} covers the constant-\(\bar{k}\) regime). For the tight and moderate scenarios at \(\hmax \in \{10, 20\}\)\,km, GSHAC overtakes the dense baseline between \(n = 1{,}000\) and \(n = 5{,}000\). At \(\hmax = 50\)\,km the sparse graph is dense enough that most of the wall-time advantage disappears. At \(n = 25{,}000\) the speedup falls to \(1.10\times\) (tight) and \(0.91\times\) (moderate), and in the loose scenario GSHAC is slower (\(0.50\times\)). The memory advantage remains in every configuration. For \(n \geq 50{,}000\) the dense baseline exceeds available memory and is omitted, while GSHAC continues to scale.

At \(n = 25{,}000\), the largest size at which the dense baseline runs on our hardware, the largest time speedup is \(5.56\times\) and the largest memory reduction \(20.2\times\), both for the tight scenario at \(\hmax = 10\)\,km. In the adversarial case at the same size (loose, \(\hmax = 50\)\,km, \(\bar{k} \approx 1{,}775\)), GSHAC still uses \(3.25\times\) less memory, although it runs \(2\times\) slower (\(0.50\times\) speedup); here the sparse graph approaches the complete graph. For \(n \geq 50{,}000\) no dense comparison is possible on this workstation. The baseline holds the full \(n \times n\) matrix (\(\approx\!18.6\)\,GiB) and its condensed copy (\(\approx\!9.3\)\,GiB) at once, about 27.9\,GiB in total, which is consistent with its measured peak of 7.0\,GiB at \(n = 25{,}000\) (Table~\ref{tab:constk}) growing quadratically and leaves no headroom on the 30\,GiB workstation. For larger \(n\), ``the dense matrix'' means the full \(n \times n\) matrix the baseline would allocate; where we quote the condensed size instead, it is the least that any dense HAC must store.

Table~\ref{tab:summary} summarises the speedups at \(n \in \{10{,}000,\,25{,}000\}\). Its ``Graph\%'' column shows that graph construction and HAC take comparable shares of the sparse runtime (42--53\%) at these sizes. On GeoNames (Section~\ref{sec:geonames}), HAC dominates when components are many and small, and graph construction dominates at high edge counts. The split also matters beyond HAC because the sparse distance graph can be reused as input to other algorithms, such as spectral clustering, MST computation, or network analysis.

\begin{figure*}[t]
  \centering
  \includegraphics[width=\textwidth]{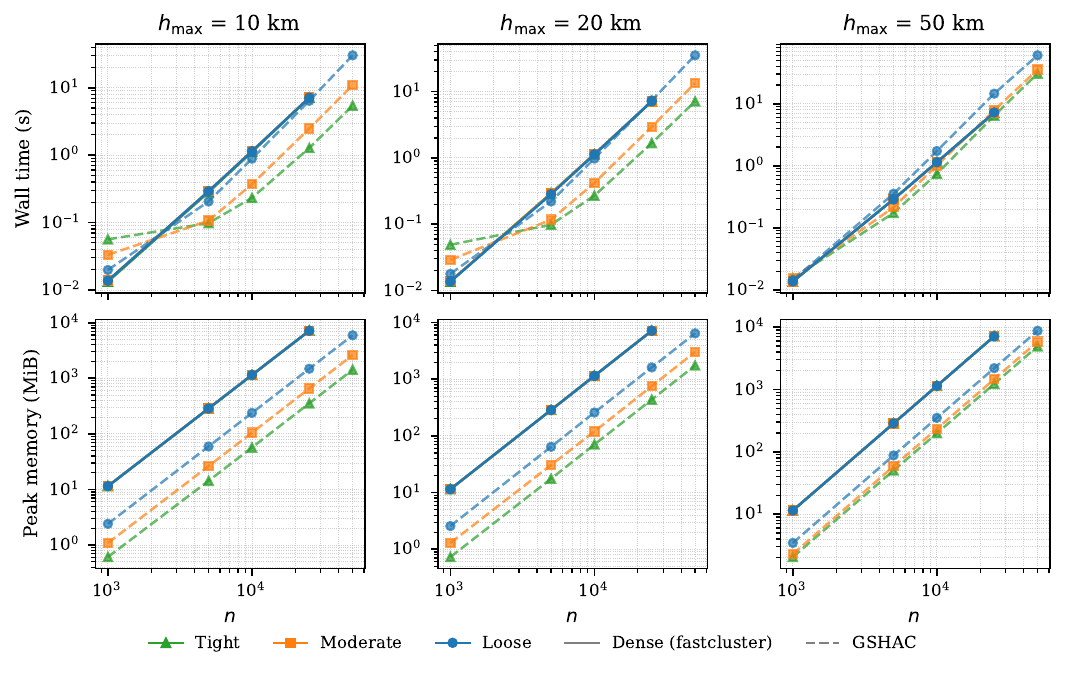}
  \caption{Fixed-domain scaling: dense (\texttt{fastcluster}, solid) and GSHAC (dashed) across three density scenarios and three~\(\hmax\) values. \emph{Top row}: wall time; \emph{bottom row}: peak memory. Because the domain is fixed, \(\bar{k}\) grows with~\(n\) and both approaches scale quadratically in this setup. GSHAC uses less memory in every panel, and less time except at small~\(n\) and in the densest graphs. GSHAC's near-linear scaling under constant~\(\bar{k}\) is shown in Figure~\ref{fig:constk}. Dense is omitted for \(n \geq 50{,}000\), where it exceeds memory.}
  \Description{A grid of line charts with logarithmic axes, one column per density scenario. In each panel, solid lines for the dense baseline and dashed lines for GSHAC rise with the number of points; in the tight and moderate scenarios at the smaller thresholds the dashed GSHAC time lines drop below the solid dense lines beyond a few thousand points, while in the loose scenario and at the largest threshold they remain above; the dashed memory lines lie below the dense lines in every panel. The dense lines stop at 25,000 points while the GSHAC lines continue to 50,000.}
  \label{fig:scaling}
\end{figure*}

\begin{table}[ht]
\caption{Speedup across density scenarios.
  ``Graph\%'' is the fraction of sparse runtime spent on graph
  construction.}
\label{tab:summary}
\small
\begin{tabular}{rrr rrr}
\toprule
$n$ & $\hmax$ & Scenario & Time & Mem & Graph\% \\
\midrule
 \multirow{3}{*}{10{,}000} & \multirow{3}{*}{10 km} & Tight & $4.82\times$ & $20.2\times$ & 42\% \\
  &  & Moderate & $3.10\times$ & $10.9\times$ & 50\% \\
  &  & Loose & $1.29\times$ & $4.81\times$ & 52\% \\
\addlinespace
 \multirow{3}{*}{10{,}000} & \multirow{3}{*}{20 km} & Tight & $4.19\times$ & $16.2\times$ & 45\% \\
  &  & Moderate & $2.72\times$ & $9.55\times$ & 50\% \\
  &  & Loose & $1.16\times$ & $4.43\times$ & 51\% \\
\addlinespace
 \multirow{3}{*}{25{,}000} & \multirow{3}{*}{10 km} & Tight & $\mathbf{5.56\times}$ & $\mathbf{20.2\times}$ & 53\% \\
  &  & Moderate & $2.87\times$ & $10.9\times$ & 52\% \\
  &  & Loose & $1.11\times$ & $4.82\times$ & 46\% \\
\addlinespace
 \multirow{3}{*}{25{,}000} & \multirow{3}{*}{20 km} & Tight & $4.27\times$ & $16.4\times$ & 53\% \\
  &  & Moderate & $2.42\times$ & $9.48\times$ & 50\% \\
  &  & Loose & $0.97\times$ & $4.43\times$ & 44\% \\
\bottomrule
\end{tabular}
\end{table}

\subsection{Scaling where dense HAC is infeasible (RQ3)}
\label{sec:realworld}

Beyond \(n \approx 50{,}000\) the dense baseline no longer fits in the workstation's memory (Section~\ref{sec:scaling}). What remains to show is whether GSHAC alone keeps scaling. Section~\ref{sec:constk} shows near-linear scaling under constant density on synthetic data. Section~\ref{sec:geonames} runs GSHAC on a two-million-point GeoNames sample whose dense matrix would need tens of terabytes, and Section~\ref{sec:geoscaling} shows both regimes of Section~\ref{sec:complexity} on real data, including how GSHAC degrades in the adversarial limit.

\subsubsection{Near-linear scaling under constant density}
\label{sec:constk}

Under constant density the mean neighbourhood size \(\bar{k}\) stays bounded as \(n\) grows (Section~\ref{sec:complexity}). The sparse graph then holds \(\Onotation{n\bar{k}} = \Onotation{n}\) edges, and the \(\Onotation{n\bar{k}\log n}\) time and \(\Onotation{n\bar{k}}\) memory become near-linear in \(n\). Geographic inventories typically grow this way, and the fixed-domain benchmark of Section~\ref{sec:scaling} does not test this regime.

To test this, we ran a separate experiment with uniformly distributed points in a domain whose side length \(L\) grows as \(L \propto \sqrt{n}\), which keeps the mean neighbourhood size \(\bar{k}\) constant across \(n \in \{1\text{K},\ldots,500\text{K}\}\). The value chosen for \(\bar{k}\) sets the absolute level of the measured curves. To check that the scaling \emph{slopes} come from the algorithm and not from this choice, we swept \(\bar{k}_{\text{target}} \in \{20, 50, 126, 200\}\) and confirmed that the resulting curves are near-parallel on a log-log plot (Figure~\ref{fig:constk}). The primary configuration uses \(\bar{k}_{\text{target}} = 126\), which matches the empirical mean neighbourhood size of GeoNames-2M at \(\hmax = 20\)\,km (Table~\ref{tab:constk}), and its curve agrees with the theoretical prediction. Memory grows linearly, from 5\,MiB at \(n\!=\!1\text{K}\) to 3.0\,GiB at \(n\!=\!500\text{K}\) (log-log slope~\(\approx 1.0\), the same at every \(\bar{k}\) tested). Time grows near-linearly, from 0.02\,s to 20\,s (log-log slope 1.14 at the primary \(\bar{k}\!=\!126\) and 1.07--1.15 across the four \(\bar{k}\) values). The advantage over \texttt{fastcluster} increases with~\(n\) and is larger at lower density. At \(n\!=\!25\text{K}\) and \(\bar{k}\!=\!126\), GSHAC is \(10.8\times\) faster and uses \(47\times\) less memory; at \(\bar{k}\!=\!50\) the two ratios reach \(27\times\) and \(118\times\). At \(n\!=\!500\text{K}\) and \(\bar{k}\!=\!126\), GSHAC uses 3.0\,GiB, whereas the condensed distance matrix alone would need \(\approx\!0.9\,\text{TiB}\).

This experiment separates the algorithm's scaling from the effect of a growing \(\bar{k}\) in a fixed domain. Across the \(\bar{k}\) sweep the slopes barely change (Figure~\ref{fig:constk}). The exponents belong to the algorithm, and \(\bar{k}_{\text{target}}\) only shifts the curves vertically. A small drift in slope remains (1.07 to 1.15 in time and 1.01 to 1.03 in memory over the tenfold range of \(\bar{k}_{\text{target}}\)). We attribute it to cache and memory-hierarchy effects, since the larger per-component sparse-matrix and MST buffers at higher \(\bar{k}\) cause more cache misses, and the Big-O analysis ignores such effects. The drift is small and does not change the near-linear scaling in \(n\) for fixed \(\bar{k}\). The GeoNames values of \(\bar{k}\) fall within the swept range (from \(\approx\!1\) at \(\hmax\!=\!1\)\,km to 126 at \(\hmax\!=\!20\)\,km), limited by the spatial distribution of the data.

\begin{table}[t]
\caption{Constant-\(\bar{k}\) scaling at the primary
  configuration (uniform points, \(\hmax = 20\)\,km,
  \(\bar{k} \approx 126\)). Domain grows as
  \(L \propto \sqrt{n}\). Dense infeasible for
  \(n \geq 50\text{K}\). Slope-invariance across
  \(\bar{k}_{\text{target}} \in \{20, 50, 126, 200\}\)
  is shown in Figure~\ref{fig:constk}.}
\label{tab:constk}
\small
\begin{tabular}{rrrrrrr}
\toprule
\(n\) & \(\bar{k}\) & \multicolumn{2}{c}{Time (s)} & \multicolumn{2}{c}{Memory (MiB)} & Time \\
\cmidrule(lr){3-4} \cmidrule(lr){5-6}
 &  & Dense & GSHAC & Dense & GSHAC & speedup \\
\midrule
1{,}000 & 105 & 0.011 & 0.018 & 11 & 5 & $0.61\times$ \\
5{,}000 & 116 & 0.285 & 0.107 & 286 & 29 & $2.68\times$ \\
10{,}000 & 119 & 1.128 & 0.235 & 1{,}145 & 59 & $4.81\times$ \\
25{,}000 & 122 & 7.040 & 0.655 & 7{,}153 & 151 & $\mathbf{10.8\times}$ \\
\addlinespace
50{,}000 & 123 & --- & 1.558 & --- & 305 & --- \\
100{,}000 & 124 & --- & 3.490 & --- & 615 & --- \\
250{,}000 & 125 & --- & 9.506 & --- & 1{,}549 & --- \\
500{,}000 & 125 & --- & 20.196 & --- & 3{,}106 & --- \\
\bottomrule
\end{tabular}
\end{table}

\begin{figure*}[t]
  \centering
  \includegraphics[width=\textwidth]{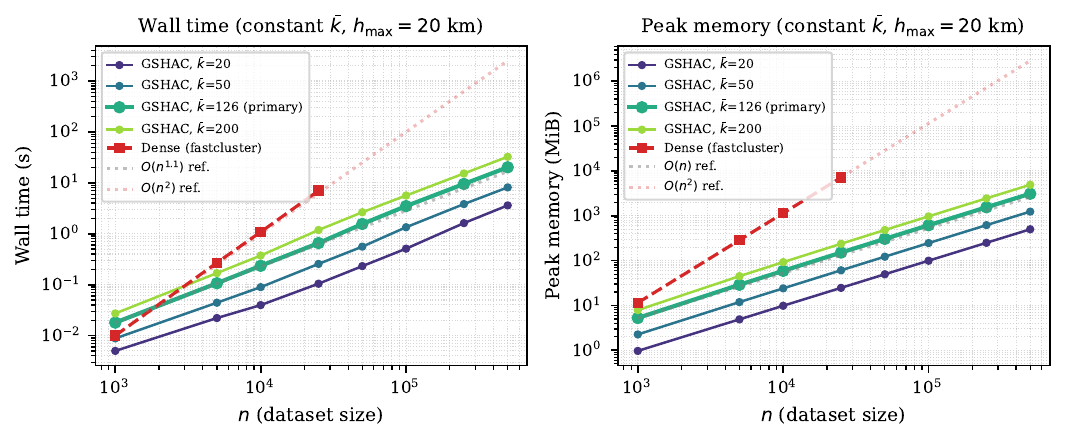}
  \caption{Constant-\(\bar{k}\) scaling at \(\hmax = 20\)\,km, swept across four target neighbourhood sizes \(\bar{k}_{\text{target}} \in \{20, 50, 126, 200\}\). \emph{Left}: wall time; \emph{right}: peak memory. The four GSHAC curves are near-parallel on log-log axes. Their fitted slopes drift upward slightly as \(\bar{k}\) grows (1.01--1.03 in memory and 1.07--1.15 in time across the tenfold range in \(\bar{k}_{\text{target}}\)), which we attribute to memory-hierarchy effects that the Big-O analysis ignores. The primary curve at \(\bar{k}_{\text{target}} = 126\) (matching GeoNames-2M at \(\hmax = 20\)\,km) is drawn thicker. The dense (\texttt{fastcluster}) baseline does not depend on \(\bar{k}\), so only the \(\bar{k}_{\text{target}} = 20\) dense run is shown; the other dense curves overlap it.}
  \Description{Two log-log line charts. In each, four nearly parallel GSHAC lines, one per neighbourhood-size setting, rise steadily with the number of points from one thousand to five hundred thousand, with the thicker primary line near the middle of the band. A steeper dense-baseline line crosses the GSHAC band and terminates at 25,000 points, beyond which only the GSHAC lines continue.}
  \label{fig:constk}
\end{figure*}

\subsubsection{Feasibility at two million points}
\label{sec:geonames}

Table~\ref{tab:geonames} summarises the results for the 2M-point GeoNames sample at four distance thresholds. All four run on the same 30\,GiB workstation, with total runtimes of 1.5 to 2.3\,minutes and whole-pipeline peak memory between 534\,MiB and 22.2\,GiB, set by graph construction. The sparse graph needs at most 1.9\,GiB (at \(\hmax = 20\)\,km), an \(8{,}000\)--\(16{,}000\times\) reduction relative to the dense distance matrix (\(\approx\!15\)\,TiB in the condensed form an optimised dense HAC requires, \(\approx\!29\)\,TiB as a full square matrix).

The MST path matters most at the larger thresholds. The largest connected component grows from 88{,}624 features at \(\hmax = 5\)\,km (a hypothetical dense sub-matrix of 58.5\,GiB) to 397{,}264 at \(\hmax = 10\)\,km (1.15\,TiB) and 1{,}035{,}840 at \(\hmax = 20\)\,km. At 20\,km, more than half of the sampled features sit in one component, whose dense sub-matrix would take 7.8\,TiB. The MST-based single-linkage path never materialises any of these matrices. Peak HAC memory stays at 438\,MiB, 1.1\,GiB, and 5.7\,GiB respectively, reductions of \({\sim}137\times\), \({\sim}1{,}100\times\), and \({\sim}1{,}410\times\) over the dense alternative. Comparing the analytical \emph{Max sub-matrix} row of Table~\ref{tab:geonames} with the measured \emph{Peak memory (HAC)} row shows why GSHAC stays feasible as components grow. Under single linkage, the MST path reduces per-component HAC memory from \(\Onotation{c_k^2}\) to \(\Onotation{c_k + m_k}\).

\begin{table}[t]
\caption{GSHAC results on a 2M-point GeoNames sample
  (\(n_{\text{full}} = 13{,}412{,}837\)), single linkage, haversine,
  mean over 5 repetitions.
  All rows are measured runs except `Max sub-matrix', the analytical dense size
  the MST path avoids materialising; the dense \(n^2\) baseline
  (\(29\,\text{TiB}\) full /
  \(15\,\text{TiB}\) condensed at 2M) is analytical and
  not run. \emph{Peak memory (graph)} is the binding peak; the MST path keeps
  \emph{Peak memory (HAC)} far below it, and the 20\,km graph approaches the
  workstation's memory limit.}
\label{tab:geonames}
\small
\begin{tabular}{lrrrr}
\toprule
& \multicolumn{4}{c}{\(\hmax\)} \\
\cmidrule(l){2-5}
& \textbf{1\,km} & \textbf{5\,km} & \textbf{10\,km} & \textbf{20\,km} \\
\midrule
Edges & 793K & 11.5M & 37.8M & 125.8M \\
Sparse storage & 12\,MiB & 176\,MiB & 577\,MiB & 1.9\,GiB \\
Components & 1{,}594{,}160 & 407{,}350 & 114{,}716 & 23{,}854 \\
Singletons & 87.4\% & 63.8\% & 58.4\% & 54.4\% \\
Largest component & 1{,}414 & 88{,}624 & 397{,}264 & 1{,}035{,}840 \\
Max sub-matrix & 0.01\,GiB & 58.5\,GiB & 1.15\,TiB & 7.8\,TiB \\
\addlinespace
Graph time & 28\,s & 33\,s & 44\,s & 74\,s \\
HAC time & 93\,s & 105\,s & 47\,s & 61\,s \\
Peak memory (graph) & 534\,MiB & 2.4\,GiB & 6.9\,GiB & 22.2\,GiB \\
Peak memory (HAC) & 362\,MiB & 438\,MiB & 1.1\,GiB & 5.7\,GiB \\
\textbf{Total time} & \textbf{2.0\,min} & \textbf{2.3\,min} & \textbf{1.5\,min} & \textbf{2.2\,min} \\
\bottomrule
\end{tabular}
\end{table}

Runtime is not monotonic in \(\hmax\). As components merge, the number of non-singleton components falls (201K \(\to\) 148K \(\to\) 48K \(\to\) 11K at 1\,/\,5\,/\,10\,/\,20\,km). HAC time is driven by per-component overhead (sub-matrix extraction, MST, union-find) and follows this count more closely than the edge count (93\,/\,105\,/\,47\,/\,61\,s). Graph construction time grows monotonically with the edge count (28\,/\,33\,/\,44\,/\,74\,s), and the total times are \(2.0\,/\,2.3\,/\,1.5\,/\,2.2\)\,min. The 10\,km threshold is the fastest overall. At 20\,km, the 125.8M-edge graph becomes the most expensive, and the bottleneck shifts from HAC to graph construction, as the complexity analysis predicts.

\subsubsection{Both growth regimes on real data}
\label{sec:geoscaling}

Section~\ref{sec:constk} established near-linear scaling on synthetic data at constant \(\bar{k}\). We now confirm it on real geographic data and contrast it with the adversarial regime, using the same 2M GeoNames sample at \(\hmax = 10\)\,km grown two different ways (Figure~\ref{fig:geoscaling}).

\paragraph{Constant density (the realistic regime).} Real inventories grow by covering \emph{more area} at similar density (national \(\to\) continental \(\to\) global). We emulate this with geographic windows of the sample, using a widening longitude strip, so that the local sampling rate and hence \(\bar{k}\) stay bounded (\(\bar{k}\) between 38 and 54 as \(n\) grows from 100K to 2M). In this regime, GSHAC scales near-linearly, with time growing as \(n^{1.14}\) and memory as \(n^{0.98}\), and it clusters the full 2M points in 91\,s using 6.9\,GiB. Since spatial datasets usually grow this way, this scalability result is highly relevant for practitioners.

\paragraph{Growing density (the adversarial limit).} If instead \(n\) grows by sampling more densely within a \emph{fixed} extent, density and therefore \(\bar{k}\) grow in proportion to \(n\) (\(\bar{k}\) from 1.9 at 100K to 37.8 at 2M). The graph then stops being sparse. The edge count grows as \(\Onotation{n^2}\), and memory approaches the quadratic behaviour of a dense method (measured slope \(n^{1.71}\), tending to \(n^2\)); time stays at \(\sim\!n^{1.0}\). Over this range \(\bar{k}\) remains small (at most 38), so the per-point costs of the range queries and of the component loop still outweigh the \(\Onotation{m}\) per-edge cost. At larger \(n\) the edge term would take over and time would also tend to \(n^2\). This is the worst case for GSHAC, and it still completes. At \(n = 2\text{M}\) it needs 6.9\,GiB on the 30\,GiB workstation, while the dense matrix (\(\approx\!29\)\,TiB) and even its condensed form (\(\approx\!15\)\,TiB) are orders of magnitude larger than any workstation's memory. Because the distance threshold bounds what is ever stored, GSHAC remains feasible in \emph{both} regimes, and as density grows its cost rises gradually from near-linear towards quadratic. Graph-construction wall time at a given \(n\) also depends on the memory locality of the input order (the GeoNames experiments traverse the sample in different orders), so absolute times should be compared within experiments, not across them. Memory and scaling exponents are unaffected.

\begin{figure}[!ht]
  \centering
  \includegraphics[width=\columnwidth]{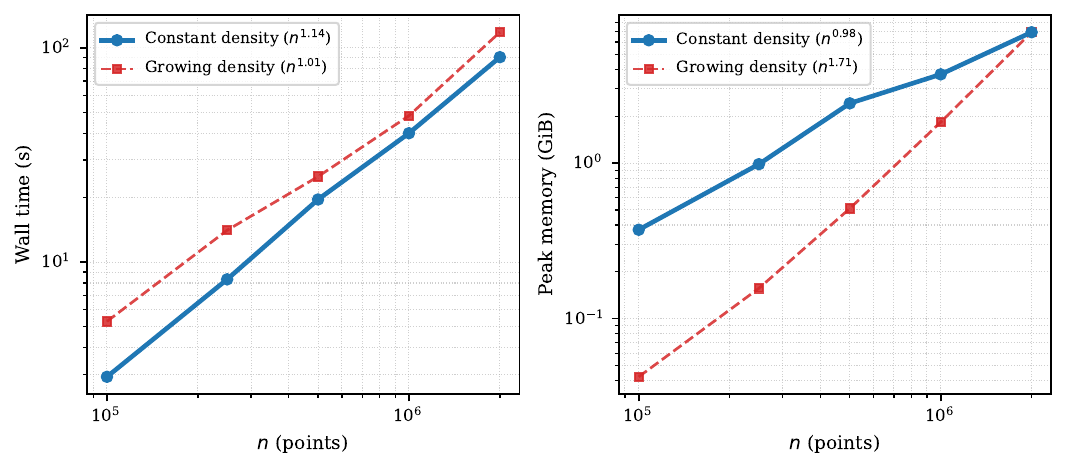}
  \caption{Two GeoNames scaling regimes on the same 2M sample at \(\hmax = 10\)\,km. \emph{Constant density} (blue, geographic windows): \(\bar{k}\) bounded, near-linear time (\(n^{1.14}\)) and memory (\(n^{0.98}\)), the realistic regime in which inventories grow by covering more area. \emph{Growing density} (red, fixed-extent subsampling): \(\bar{k} \propto n\), memory converging to quadratic (\(n^{1.71}\)). The memory curves meet at \(n = 2\text{M}\), where both equal the full sample (wall times differ there because the two experiments traverse the sample in different orders; Section~\ref{sec:geoscaling}). Even in the adversarial regime GSHAC stays within the 30\,GiB workstation (6.9\,GiB), while the dense matrix would need \(\approx\!29\)\,TiB.}
  \Description{Two log-log line charts of time and memory versus number of points for the GeoNames sample. A blue line for the constant-density regime rises with a near-unit slope, and a red line for the growing-density regime rises more steeply in the memory panel. The two lines converge at two million points.}
  \label{fig:geoscaling}
\end{figure}

Figure~\ref{fig:limits} draws together the constant-\(\bar{k}\) scaling of Section~\ref{sec:constk}, the point where the dense baseline exhausts memory, and the two-million-point GeoNames benchmark of Section~\ref{sec:geonames}. It plots the measured GSHAC curves at the primary configuration \(\bar{k}_{\text{target}} = 126\) (the value that matches GeoNames-2M at \(\hmax = 20\)\,km) with their power-law extrapolation, alongside the \(\Onotation{n^2}\) dense baseline and the measured GeoNames-2M point (green star). The dense baseline reaches the 30\,GiB limit near \(n \approx 50{,}000\) and cannot run past it, whereas GSHAC continues along a \(\sim\!n^{1.1}\) time and \(\sim\!n^{1.0}\) memory trend into the tens of millions. The measured GeoNames-2M point falls close to the extrapolated GSHAC curve and well below the dense baseline, within the 30\,GiB budget, so the projection from synthetic constant-density data holds on a real dataset; the constant-density geographic experiment of Section~\ref{sec:geoscaling} shows the same near-linear trend measured directly on GeoNames. Exact HAC therefore stays feasible two to three orders of magnitude beyond the dataset size at which the dense baseline fails.

\begin{figure*}[t]
  \centering
  \includegraphics[width=\textwidth]{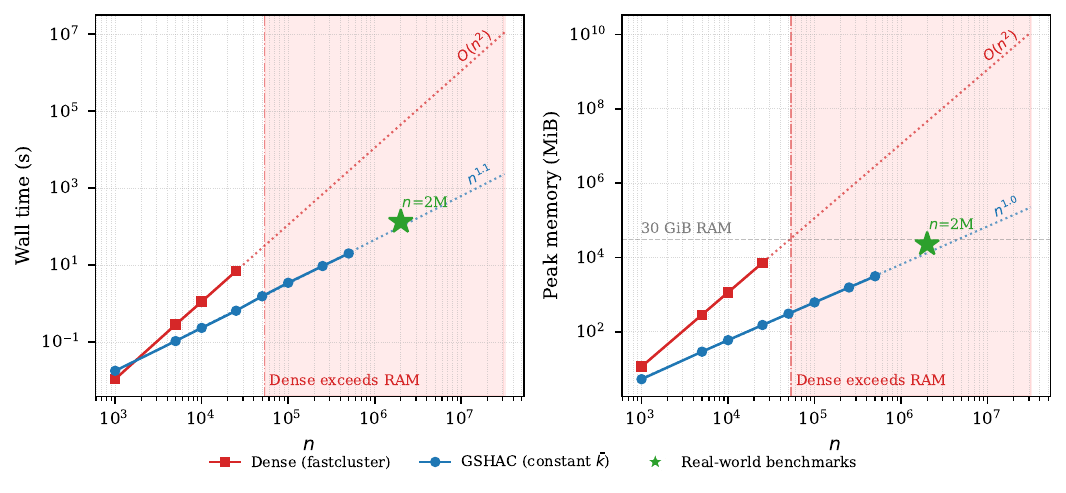}
  \caption{Scalability projection at \(\hmax = 20\)\,km. The blue line is the constant-\(\bar{k}\) experiment at the primary configuration \(\bar{k}_{\text{target}} = 126\), matching GeoNames-2M at \(\hmax = 20\)\,km; slope-invariance across other \(\bar{k}\) values is shown in Figure~\ref{fig:constk}. \emph{Left}: wall time; \emph{right}: peak memory. Measured data (solid) and power-law extrapolation (dotted). GSHAC scales as \(\sim\!n^{1.1}\) in time and \(\sim\!n^{1.0}\) in memory; the dense baseline (red) scales as \(\Onotation{n^2}\) and exceeds 30\,GiB beyond \(n \approx 50\text{K}\). The green star marks the GeoNames-2M benchmark at \(\hmax = 20\)\,km.}
  \Description{Two log-log charts of time and memory versus number of points. A blue GSHAC line rises gently and extends as a dotted extrapolation into the millions of points; a red dense-baseline line rises much more steeply and crosses a horizontal 30 GiB memory-limit line near fifty thousand points. A green star sits on the extrapolated blue line at two million points.}
  \label{fig:limits}
\end{figure*}

\subsection{Positioning against existing tools (RQ4)}
\label{sec:positioning}

%PySAL's \texttt{DistanceBand} builds only the graph. scikit-learn's connectivity HAC computes the full clustering, but more slowly and, for some linkages, incorrectly. HDBSCAN optimises a different, density-based objective.

\subsubsection{Graph construction versus PySAL}

\begin{sloppypar}
Table~\ref{tab:pysal} compares graph construction time with PySAL's \texttt{DistanceBand}. GSHAC uses \texttt{cKDTree\allowbreak .query\_pairs}, a pure-C routine that returns an array, whereas \texttt{DistanceBand} calls \texttt{cKDTree.\allowbreak sparse\_distance\_matrix} to build a Python dictionary-of-keys object. GSHAC is 43--80\(\times\) faster across all scenarios and sizes in Table~\ref{tab:pysal}, and the speedup is roughly constant from \(n = 5{,}000\) onwards.
\end{sloppypar}

\begin{table}[ht]
  \centering
  \caption{Graph-construction time: GSHAC \texttt{query\_pairs}
           vs.\ PySAL \texttt{DistanceBand}, \(h_{\max}=10\)\,km.}
  \label{tab:pysal}
  \small
  \begin{tabular}{lrrrr}
    \toprule
    Scenario & \(n\) & GSHAC (s) & DistanceBand (s) & Ratio \\
    \midrule
    Tight & 1{,}000 & 0.001 & 0.048 & 43$\times$ \\
    Tight & 5{,}000 & 0.015 & 1.153 & 79$\times$ \\
    Tight & 10{,}000 & 0.062 & 4.579 & 74$\times$ \\
    Tight & 25{,}000 & 0.400 & 29.700 & 74$\times$ \\
    \addlinespace
    Moderate & 1{,}000 & 0.002 & 0.082 & 54$\times$ \\
    Moderate & 5{,}000 & 0.027 & 2.157 & 80$\times$ \\
    Moderate & 10{,}000 & 0.117 & 8.403 & 72$\times$ \\
    Moderate & 25{,}000 & 0.749 & 54.732 & 73$\times$ \\
    \addlinespace
    Loose & 1{,}000 & 0.003 & 0.176 & 65$\times$ \\
    Loose & 5{,}000 & 0.066 & 4.687 & 71$\times$ \\
    Loose & 10{,}000 & 0.271 & 19.102 & 70$\times$ \\
    Loose & 25{,}000 & 1.771 & 129.127 & 73$\times$ \\
    \bottomrule
  \end{tabular}
\end{table}

\subsubsection{Comparison with scikit-learn connectivity HAC}
\label{sec:sklearn_comparison}

\begin{sloppypar}
Table~\ref{tab:sklearn} compares GSHAC with scikit-learn's \texttt{Agglomerative\allowbreak Clustering} using sparse \texttt{connectivity} and single linkage (moderate scenario, \(\hmax = 10\)\,km, \(h = 5\)\,km). Both return identical partitions (ARI\,=\,1.0) in every configuration. Timed end to end (sparse-graph construction plus clustering, against scikit-learn's connectivity build plus \texttt{fit}), GSHAC is 9--12\(\times\) faster across the tested sizes.
\end{sloppypar}

The gap has two causes. First, scikit-learn's \texttt{\_fix\_connectivity} computes inter-component pairwise distances to connect disconnected components before HAC, whereas GSHAC treats disconnected components as independent subproblems. Second, scikit-learn recomputes distances during HAC because its \texttt{connectivity} parameter accepts only binary matrices; our sparse graph already stores the pre-computed distances.

\begin{table}[t]
\caption{GSHAC vs.\ scikit-learn connectivity HAC (single linkage,
  moderate, \(\hmax = 10\)\,km, \(h = 5\)\,km), timed
  \emph{end-to-end} (sparse-structure construction + clustering) so both
  sides are measured symmetrically. Both recover the identical partition
  (ARI\,=\,1.0).}
\label{tab:sklearn}
\small
\begin{tabular}{r r r r r}
\toprule
\(n\) & GSHAC (s) & scikit-learn (s) & Ratio & Clusters \\
\midrule
500 & 0.022 & 0.250 & 11.2$\times$ & 49 \\
1{,}000 & 0.026 & 0.290 & 11.0$\times$ & 49 \\
2{,}500 & 0.039 & 0.456 & 11.7$\times$ & 48 \\
5{,}000 & 0.087 & 1.062 & 12.2$\times$ & 49 \\
10{,}000 & 0.305 & 3.287 & 10.8$\times$ & 48 \\
25{,}000 & 2.020 & 18.232 & 9.0$\times$ & 48 \\
\bottomrule
\end{tabular}
\end{table}

\subsubsection{Scope comparison with HDBSCAN}
\label{sec:hdbscan}

HDBSCAN~\cite{mcinnes2017hdbscan} is the most widely used density-based clustering algorithm for large spatial datasets, and practitioners often use it in place of exact HAC. We therefore compare the resource profiles of HDBSCAN and GSHAC to show the difference in scope. The two algorithms solve different problems and their outputs are not interchangeable, so this is not a benchmark of equivalent methods. We use a single implementation (\texttt{sklearn.cluster.\allowbreak HDBSCAN} v1.8).

HDBSCAN builds a mutual-reachability MST and derives a density-based cluster tree from it. Its clusters reflect density instead of a distance threshold. Depending on \texttt{min\_cluster\_size} (hereafter \texttt{mcs}), 12.7\% to 35.2\% of features may be labelled as noise, controlled by \texttt{mcs}, not by a geographic distance. Complete, average, and Ward linkage are not available.

Table~\ref{tab:hdbscan} reports timing, memory, and cluster counts on a 100K subset of the GeoNames sample (haversine, \(\hmax = 10\)\,km), measured for both methods in the same harness. GSHAC finishes in 5.6\,s, 39--41\(\times\) faster than either HDBSCAN run (229\,s at \texttt{mcs}\,=\,2 and 219\,s at \texttt{mcs}\,=\,5), with comparable peak memory. HDBSCAN's runtime hardly depends on \texttt{mcs} (229\,s vs 219\,s), so no setting makes it competitive in time. The gap widens with size, since scikit-learn's HDBSCAN scales super-linearly on haversine input (log-log slope~\(\approx 1.9\) in a scaling probe over \(n \in \{10\text{K},\ldots,100\text{K}\}\)), which extrapolates to \(\approx\!18\,\text{h}\) on the full 2M sample, against minutes for GSHAC (Section~\ref{sec:geonames}).

\begin{table}[!ht]
\caption{Resource comparison: scikit-learn HDBSCAN vs.\ GSHAC on a 100K
  subset of the GeoNames sample (\(n = 100{,}000\), haversine), both measured
  in this run with the same whole-pipeline peak-memory definition.
  Non-equivalent outputs; illustrates resource profiles only. The gap widens with
  \(n\): HDBSCAN scales super-linearly (extrapolated to hours at 2M), GSHAC
  near-linearly (Table~\ref{tab:geonames}).}
\label{tab:hdbscan}
\small
\begin{tabular}{lrrrr}
\toprule
\textbf{Method} & \textbf{Time} & \textbf{Peak mem} & \textbf{Clusters} & \textbf{Noise} \\
\midrule
HDBSCAN (mcs\,=\,2) & 229\,s & 43\,MiB & 29{,}398 & 12.7\% \\
HDBSCAN (mcs\,=\,5) & 219\,s & 33\,MiB & 4{,}667 & 35.2\% \\
\addlinespace
GSHAC (\(\hmax\!=\!10\)\,km)\(^{*}\) & 5.6\,s & 35\,MiB & --- & 0\% \\
\quad\(\to h\!=\!5\)\,km  & & & 80{,}571 & \\
\quad\(\to h\!=\!10\)\,km & & & 56{,}141 & \\
\bottomrule
\multicolumn{5}{l}{\footnotesize \(^{*}\)Both: mean over 5 repetitions. One dendrogram cut at any \(h \leq \hmax\).} \\
\multicolumn{5}{l}{\footnotesize GSHAC peak is whole-pipeline: graph phase 35\,MiB, MST-path HAC phase 24\,MiB.}
\end{tabular}
\end{table}

HDBSCAN labels 12.7\% (\texttt{mcs}\,=\,2) to 35.2\% (\texttt{mcs}\,=\,5) of the features as noise, while GSHAC assigns every feature to a cluster. In GSHAC, the cluster count follows from one interpretable parameter, the cut height~\(h\). A single 5.6\,s run produces a full dendrogram that can be cut at any \(h \le \hmax\) (e.g.\ 56{,}141 clusters at 10\,km and 80{,}571 at 5\,km). HDBSCAN needs a separate run for each \texttt{mcs} and cannot produce standard linkage dendrograms.

% =======================================================================
\section{Other linkage methods}
\label{sec:otherlinkage}

So far the evaluation has used single linkage, which the MST path scales to the largest datasets (Section~\ref{sec:linkage-scaling}). Exactness is not specific to single linkage, though. Theorem~\ref{thm:exact} also covers complete and average linkage, and the multi-linkage check of Section~\ref{sec:exactness} reproduced the dense \texttt{fastcluster} partition for all four methods, including Ward. The linkages differ in what it costs to compute the result.

Complete, average, and Ward linkage have no MST path. Their merge heights depend on intra-component distances beyond \(\hmax\), so each connected component's full \(c_k \times c_k\) sub-matrix is recomputed from coordinates and passed to a dense routine. Peak memory is then \(\Onotation{\max_k c_k^2}\), compared with \(\Onotation{c_k + m_k}\) for single linkage, and feasibility depends on the largest component, not on \(n\) or the mean neighbourhood size. On GeoNames this limit comes early. At \(\hmax = 5\)\,km the largest component holds 88{,}624 features, a 58.5\,GiB sub-matrix that already exceeds the 30\,GiB workstation, and at \(\hmax \ge 10\)\,km it passes 1\,TiB, so only single linkage runs at these thresholds. When components stay small, these linkages are practical as well, and in the multi-linkage configurations of Section~\ref{sec:exactness} they match the dense partition exactly.

Ward carries the further caveat of Section~\ref{sec:correctness}. Its merge heights never decrease, but they are not bounded below by the minimum inter-point distance, so the decomposition can fail for Ward when one component encloses another at a separation above \(\hmax\). This does not happen for spatially compact geographic clusters, and dense Ward reproduced in every case we tested, but for Ward the guarantee is empirical, not proven. These properties suggest how to choose. Single linkage suits the largest datasets and distance-threshold cluster definitions because the MST path removes the limit on component size. Complete and average linkage give the same exactness whenever the largest component fits in memory, which covers regional datasets and any sample whose \(\hmax\) keeps components small. Ward suits compact point clusters when a minimum-variance criterion is wanted. The implementation refuses Ward for non-point geometries, whose boundary-to-boundary distances have no Euclidean embedding for the variance criterion. Single, complete, and average linkage depend only on the pairwise dissimilarities and so apply to these distances unchanged.

% =======================================================================
\section{Discussion}
\label{sec:discussion}

GSHAC helps whenever the analyst has a meaningful upper bound~\(\hmax\) on the clustering distance. Such bounds arise naturally in environmental, regulatory, and operational settings. The speedup grows with \(n/\bar{k}\), so data that cluster tightly relative to~\(\hmax\) gain the most. GSHAC helps little in two cases. When \(\hmax\) approaches the dataset's spatial extent, the graph approaches a complete graph and \(\bar{k} \approx n\). And for the linkages other than single linkage, one very large component can dominate the runtime even when the overall \(\bar{k}\) is small, which limits them to problems whose largest component fits in memory (Section~\ref{sec:otherlinkage}). At small~\(n\), GSHAC's fixed overheads (index construction and per-component processing) exceed the negligible setup cost of the dense baseline. Where GSHAC overtakes the dense baseline in our benchmarks (the tight and moderate scenarios at \(\hmax \in \{10, 20\}\)\,km), the crossover lies between \(n = 1{,}000\) and \(n = 5{,}000\) (Section~\ref{sec:scaling}). GSHAC is meant for moderate to large~\(n\). Performance depends mostly on the distribution of component sizes. In GeoNames most components are tiny (at \(\hmax = 1\)\,km the median size is~1 and the 99th percentile is~5) and are clustered in microseconds, while a few grow into very large components as \(\hmax\) increases. The MST single-linkage path keeps those large components tractable, because it reduces per-component HAC memory from quadratic to linear in component size (Section~\ref{sec:geonames}).

The comparison with HDBSCAN (Section~\ref{sec:hdbscan}) is mainly about scope. HDBSCAN clusters reflect density. When a regulation or an operational rule defines a site as ``all features within \(h\)\,km'', a density-based partition cannot stand in for it, however fast it runs. Sparse-input exact HAC has a parallel in bioinformatics~\cite{loewenstein2008efficient,nguyen2014sparsehc}, where BLAST-threshold similarity graphs play the role of our geographic distance graph. The geographic setting has one practical advantage. Metric distances allow the graph to be built efficiently, in \(\Onotation{n \log n + m}\) time with spatial indexes. Meanwhile, previously proposed bioinformatics methods require external pre-computation, such as an all-against-all BLAST search.

In the loose synthetic scenario at \(\hmax = 50\)\,km, GSHAC is slower than \texttt{fastcluster} in wall time (speedups of \(0.50\)--\(0.95\times\) across sizes, worst at \(n = 25{,}000\)). The graph is very dense in this adversarial regime (up to \(\bar{k} \approx 1{,}775\), a graph density of 7\%, at \(\hmax = 50\)\,km). The sparse representation then carries overhead from index lookups and sparse matrix operations that a contiguous condensed array avoids. A fully native implementation of the sparse-matrix extraction and the component loop would narrow this gap.

%The memory advantage follows from what is stored. GSHAC keeps \(\Onotation{n \bar{k}}\) entries instead of \(\Onotation{n^2}\), which gave \(3.25\)--\(20.2\times\) lower peak memory across all scenarios. At \(n \geq 50{,}000\), the dense baseline as implemented exceeds the 30\,GiB workstation (Section~\ref{sec:scaling}), and any dense implementation must store at least the condensed matrix, which grows quadratically to \(\approx\!15\)\,TiB at \(n = 2\text{M}\). Under bounded density, GSHAC's \(\Onotation{n \bar{k} \log n}\) time against the \(\Onotation{n^2}\) of dense HAC means that GSHAC will beat any dense implementation beyond some dataset-dependent size.

Müllner~\cite{mullner2013fastcluster} benchmarked \texttt{fastcluster} against \texttt{scipy.cluster\allowbreak .hierarchy}, R's \texttt{hclust}, MATLAB, and Mathematica on synthetic Gaussian-mixture data up to \(n \approx 20{,}000\), and showed that choosing the right algorithm (MST-based single linkage, nearest-neighbour chain for complete/average/Ward) reduces the HAC step from \(\Onotation{n^3}\) to \(\Onotation{n^2}\). Those experiments assume that the condensed distance matrix is already in memory. They therefore do not measure the cost that dominates for large geospatial datasets, namely computing and storing \(\binom{n}{2}\) distances. For the 2M GeoNames sample, the condensed matrix alone would need \(\approx 15\)\,TiB, three orders of magnitude more than workstation RAM, so \texttt{fastcluster} cannot even be called. GSHAC complements \texttt{fastcluster}. It removes the \(\Onotation{n^2}\) storage bottleneck that Müllner identifies but does not address and hands per-component HAC to a \texttt{fastcluster}-style \(\Onotation{c_k^2}\) routine on dense sub-matrices that fit in memory. When \(\bar{k}\) is bounded by data locality, the typical case for geographic data (the GeoNames experiments give \(\bar{k}\) of roughly 1--130, and slope invariance was verified across \(\bar{k}_{\text{target}} \in \{20, 50, 126, 200\}\)), this reduces the end-to-end complexity from \(\Onotation{n^2}\) to \(\Onotation{n \log n}\) and allows datasets two to three orders of magnitude larger on the same hardware (Table~\ref{tab:constk}, Figure~\ref{fig:constk}).

GSHAC needs only a spatial range query and an exact pairwise distance function. This is why the implementation also supports line, polygon, and mixed-geometry inputs in projected CRSs. A GEOS STRtree \texttt{dwithin} range query replaces the KD-tree or ball tree, and exact boundary-to-boundary minimum distances replace the point kernels, with distance 0 for touching or overlapping features as in GEOS, \texttt{sf}, and PostGIS. The component decomposition and exactness guarantee remain unchanged, since Theorem~\ref{thm:exact} requires only a complete enumeration of pairs within \(\hmax\) with exact dissimilarities. Only the range query and distance kernel differ, and the kernel costs more per pair than for points. All experiments in this paper use point features. Geodesic boundary-to-boundary distances for geometries in geographic CRSs are left for future work. Work on secondary-partitioning indexes for distance queries on non-point data~\cite{michalopoulos2025distance} shows that polygon-aware range queries scale to large datasets, which provides the indexing primitive for that extension. Non-metric dissimilarities that increase with geographic separation, such as travel time, can be handled in the same way if a spatial index over the alternative metric is available.

Because the components within \(\hmax\) are independent, the per-component HAC is embarrassingly parallel. GeoNames at \(\hmax = 10\)\,km decomposes into 114{,}716~components, 48K of them non-singleton (Section~\ref{sec:geonames}), and a na\"ive parallel map over the laptop's 8~cores could shorten the HAC phase considerably. We have not tried this yet, but GSHAC's small per-component memory footprint makes multi-core execution practical even on commodity hardware. Dense nearest-neighbour-chain linkage is inherently sequential and has no comparable decomposition, which means the single-core speedups in Section~\ref{sec:scaling} understate what GSHAC can achieve in parallel. Absolute wall-clock speedups still depend on the system. They vary with the CPU frequency governor and with memory allocation and paging (in particular the availability of transparent huge pages; Section~\ref{sec:expdesign}), so they should be read as representative, not exact. The memory reductions, log-log scaling exponents, and exactness do not depend on the hardware, and these are the core results we expect others to reproduce.

% =======================================================================
\section{Conclusion}
\label{sec:conclusion}

GSHAC makes exact hierarchical agglomerative clustering feasible for spatial datasets with millions of features. It replaces the \(\Onotation{n^2}\) dense distance matrix with an \(\Onotation{n\bar{k}}\) sparse geographic distance graph built by spatial indexing. It splits the graph through its connected components into independent subproblems that are solved exactly below the distance threshold. We proved that this gives the same cluster assignments as dense HAC for single, complete, and average linkage, and Ward was exact in practice for spatially compact clusters. The method is available as a scikit-learn-compatible \texttt{SparseAgglomerativeClustering} estimator. On a two-million-point GeoNames sample, each threshold from 1 to 20\,km clustered in under 2.5\,minutes within about 22\,GiB on a single 30\,GiB workstation, where the dense matrix (\(\approx\!29\)\,TiB) cannot be formed. Under constant density, the method scales near-linearly into the millions. To our knowledge, these are the largest demonstrations of exact hierarchical clustering of geographic data on a single workstation. In practice, this means that exact multi-scale clustering of a global inventory, which used to require downsampling or cluster hardware, now runs on a laptop. Future work includes: (i)~adapting GSHAC for streaming settings where features arrive incrementally; (ii)~extending to parallel execution for datasets beyond single-machine memory; (iii)~applying the sparse geographic graph as a general-purpose input to other distance-based spatial algorithms; and (iv)~developing memory-efficient HAC paths for complete, average, and Ward linkage analogous to the MST-based single-linkage optimisation.

% =======================================================================
\begin{acks}
Funded by the European Union. This work was supported by the European Research Council (ERC) project MINE-THE-GAP (grant agreement no. 101170578 \url{https://doi.org/10.3030/101170578}). Views and opinions expressed are however those of the author only and do not necessarily reflect those of the European Union or the European Research Council Executive Agency. Neither the European Union nor the granting authority can be held responsible for them.

The authors acknowledge the use of Claude Opus 4.8 (Anthropic) for assistance in drafting sections of the text as well as source code implementation. Additionally, the authors used Claude Opus 4.8 and Grammarly (Grammarly Inc.) to improve grammar, clarity, and readability. The authors have reviewed, verified, and audited all outputs, including code blocks, and take full responsibility for the correctness and integrity of the entire Work.

\end{acks}

\section*{Ethics and Privacy Statement}

This work presents a clustering algorithm and evaluates it on synthetic data and GeoNames, a public, openly licensed (CC~BY~4.0) database of named geographic places. No personal or sensitive data are processed, and we foresee minimal privacy risk. Spatial clustering of industrial or environmental inventories can inform regulatory assessments and land-use planning; results should be interpreted with domain context, as proximity-based cluster membership does not imply common ownership, operational links, or environmental responsibility.

% =======================================================================
\section*{Data and Code Availability}

The GSHAC library is open source at \url{https://github.com/mine-the-gap/gshac} (release~1.1.0, GPL-3.0; archived at \url{https://doi.org/10.5281/zenodo.19567090}) and provides the scikit-learn-compatible estimator used in this paper. The reproducibility is available at \url{https://github.com/mine-the-gap/gshac-reproducibility} and archived at \url{https://doi.org/10.5281/zenodo.23784622}. It contains the benchmark scripts, the data pipeline, the pinned software environment, and the generators for every figure and table. All synthetic data use a fixed random seed. The GeoNames input (CC~BY~4.0) is downloaded automatically; the analysed snapshot is dated 2026-03-25. The reported benchmarks were run in August 2026 with gshac release~1.1.0. A reduced-scale pass (\texttt{make reproduce-fast}) reproduces the paper's core findings in under an hour on a commodity machine.

% =======================================================================
\bibliographystyle{ACM-Reference-Format}
\bibliography{references}

% =======================================================================
\appendix

\section{Implementation details}
\label{sec:system}

The Python implementation relies on C libraries for all computationally intensive operations. Range queries use \texttt{scipy.\allowbreak spatial.\allowbreak cKDTree}~\cite{virtanen2020} for Euclidean distances and \texttt{sklearn.\allowbreak neighbors.\allowbreak BallTree}~\cite{pedregosa2011sklearn} with the haversine metric for geodesic distances. The sparse distance graph is a \texttt{scipy.sparse.\allowbreak csr\_matrix}, and \texttt{scipy.sparse.\allowbreak csgraph.\allowbreak connected\_\allowbreak components} performs the component decomposition.

\subsection{Input formats and CRS-aware dispatch}

The implementation accepts input coordinates either as a raw \((n, 2)\) numeric array or as a \texttt{geopandas} \texttt{GeoSeries}/\allowbreak \texttt{GeoDataFrame}, with an optional explicit \texttt{crs=} argument for the array form. When CRS information is available, GSHAC chooses the distance metric automatically from whether the CRS is geographic or projected, following the precedent of \texttt{sf::st\_distance} in R~\cite{pebesma2018sf}. Ellipsoidal geodesic distance is the default for geographic CRSs, computed via \texttt{pyproj.Geod.inv} on the WGS-84 ellipsoid, and Euclidean distance is the default for projected CRSs. The dispatch layer rejects metric/CRS combinations that would silently produce wrong results. For instance, requesting \texttt{haversine} on projected coordinates, or \texttt{euclidean} on geographic coordinates, raises an error instead of evaluating distance on the wrong manifold.

For geographic CRSs the default is the ellipsoid, not the sphere. This deliberately departs from the s2-based default of \texttt{sf}. The reason is that defaults should be correct. GSHAC output typically feeds maps, area summaries, or further analysis, where callers assume ellipsoidal correctness without checking it, and the distance error of a spherical approximation carries into later area, length, and adjacency computations. The cost is a constant-factor overhead in the distance-kernel stage of \texttt{spatial\_dist\_graph}. Empirically it is about \(1.0\times\) the C haversine kernel in sparse-graph regimes (\(\bar{k} \lesssim 10\), where the \texttt{BallTree} neighbour query dominates) and rises to \(\approx\!2.7\)--\(3.3\times\) in dense regimes (\(\bar{k}\) in the hundreds or thousands, where the distance kernel itself dominates). The asymptotic complexity is the same, since both paths are \(\Onotation{n\bar{k}}\). We consider this overhead an acceptable price for correct defaults, and users who need maximum throughput can switch back to the sphere.

\begin{sloppypar}
The spherical \emph{haversine} metric remains available as \texttt{metric="haversine"}, and we use it for the geographic benchmarks of Section~\ref{sec:results}, for two reasons. First, none of the dense-HAC implementations we compare against (\texttt{fastcluster}~\cite{mullner2013fastcluster}, \texttt{scipy.cluster.\allowbreak hierarchy}, and scikit-learn's \texttt{AgglomerativeClustering}~\cite{pedregosa2011sklearn} with sparse connectivity) has a native ellipsoidal-distance kernel. A fair comparison therefore needs both sides to use the same spherical formula on \((\text{lon}, \text{lat})\) inputs, which separates the contribution of the sparse-graph algorithm from that of the distance kernel. Second, earlier versions of the GSHAC package and the C extension \texttt{\_gshac.so} use the haversine kernel, so using it here keeps the experimental numbers consistent with that earlier work. Relative to the WGS-84 ellipsoid, the haversine distance error is at most \(\sim\!0.5\%\) (in the worst case, near the poles for short east-west chords, and an order of magnitude smaller at mid-latitudes), well below the precision that matters for clustering at the thresholds we evaluate. On a representative clustered \((\text{lon}, \text{lat})\) dataset, the geodesic path finds \(100\%\) of the haversine edges plus only \(\approx\!0.005\%\) additional boundary edges, and the numbers of connected components are identical.
\end{sloppypar}

Line, polygon, and mixed-geometry inputs are supported in projected CRSs via a GEOS-backed path, in which \texttt{shapely.\allowbreak STRtree} range queries with the \texttt{dwithin} predicate enumerate all pairs within \(\hmax\) by true boundary-to-boundary minimum distance, and per-component sub-matrices for complete and average linkage are recomputed with the same GEOS kernel. Ward is refused for non-point input, and geographic-CRS geometries raise an error with a reprojection hint (GEOS distances are planar). All benchmarks in this paper use point features.

\subsection{Two HAC paths}
\begin{sloppypar}
For \emph{single linkage}, GSHAC computes \texttt{scipy.sparse.\allowbreak csgraph.\allowbreak minimum\_\allowbreak spanning\_tree} on each component's sparse sub-matrix and then builds the linkage matrix~\(Z\) with a union-find pass over the \(c_k - 1\) MST edges. This path takes \(\Onotation{m_k \log m_k}\) time and \(\Onotation{c_k + m_k}\) memory and forms no dense matrix, however large the component. The resulting \(Z\) is a standard SciPy linkage matrix, so \texttt{fcluster} can cut it at any height \(h \le \hmax\) and it can be drawn as a full dendrogram. For the \emph{other linkage methods} (complete, average, Ward), the full per-component distance matrix is recomputed from the original coordinates with the graph's metric, which recovers the intra-component distances beyond \(\hmax\) that the threshold graph omits, as Theorem~\ref{thm:exact} requires. The matrix is then condensed with \texttt{scipy.spatial.\allowbreak distance.\allowbreak squareform} and passed to \texttt{fastcluster.\allowbreak linkage}~\cite{mullner2013fastcluster}, with \texttt{scipy.cluster.\allowbreak hierarchy.\allowbreak linkage} as a fallback when \texttt{fastcluster} is unavailable.
\end{sloppypar}

Only the orchestration layer (component iteration and label assignment) is pure Python. Spatial queries, matrix operations, and HAC run in C, through the SciPy and scikit-learn libraries for spatial indexing and matrix operations and through a dedicated C extension (\texttt{\_gshac.so}). The extension implements union-find linkage construction over MST edges, haversine distance computation, and batch threshold-based dendrogram cutting, and it releases the GIL. A pure-Python fallback is used when the C extension is not built.

\subsection{scikit-learn API}
\begin{sloppypar}
To support adoption in existing pipelines we provide two public functions. \texttt{geographic\_\allowbreak connectivity} returns a binary sparse matrix in compressed sparse row (CSR) format encoding the distance-band graph, compatible with scikit-learn's \texttt{connectivity} parameter and other graph-based spatial tools. \texttt{Sparse\allowbreak Agglomerative\allowbreak Clustering} is a full estimator inheriting \texttt{Base\allowbreak Estimator} and \texttt{Cluster\allowbreak Mixin}, mirroring the \texttt{Agglomerative\allowbreak Clustering} API. The user specifies \texttt{h\_max} (\(\hmax\)), \texttt{distance\_\allowbreak threshold} (\(h\)), \texttt{linkage}, and \texttt{metric}, then calls \texttt{fit(X)} to obtain \texttt{labels\_}, \texttt{n\_clusters\_}, and \texttt{n\_connected\_\allowbreak components\_}.
\end{sloppypar}

\end{document}